\documentclass[12pt,english]{article}
\usepackage[T1]{fontenc}
\usepackage[latin9]{inputenc}
\usepackage{geometry}
\usepackage{graphicx}
\usepackage{amssymb}
\usepackage[authoryear]{natbib}
\usepackage{rotfloat}

\usepackage{amsmath}
\usepackage{amsthm}

\makeatletter
\usepackage{epsf}\@ifundefined{definecolor}
 {\usepackage{color}}{}

\makeatletter

\providecommand{\tabularnewline}{\\}

\usepackage{epstopdf}

\newcommand{\captionfonts}{\footnotesize}\makeatletter%
\long\def\@makecaption#1#2{%
  \vskip\abovecaptionskip
  \sbox\@tempboxa{{\captionfonts #1: #2}}%
  \ifdim \wd\@tempboxa >\hsize
    {\captionfonts #1: #2\par}
  \else
    \hbox to\hsize{\hfil\box\@tempboxa\hfil}%
  \fi
  \vskip\belowcaptionskip}\makeatother%

\newtheorem{theorem}{Theorem}[section]

\newtheorem{lemma}{Lemma}[section]
\makeatother

\usepackage{babel}

\makeatother

\begin{document}

\title{\vspace{-.5in}Unbiased Instrumental Variables Estimation Under Known First-Stage Sign}

\author{Isaiah Andrews  \\
MIT
\thanks{email: iandrews@mit.edu}
\and
Timothy B. Armstrong  \\
Yale University
\thanks{email: timothy.armstrong@yale.edu.  We thank Gary Chamberlain, Jerry Hausman, Max Kasy, Frank Kleibergen, Anna Mikusheva, Daniel Pollmann,  Jim Stock, and audiences at the 2015 Econometric Society World Congress, the 2015 CEME conference on Inference in Nonstandard Problems, Boston College, Boston University, Chicago, Columbia, Harvard, Michigan State, MIT, Princeton, Stanford, U Penn, and Yale
for helpful comments, and Keisuke Hirano and Jack Porter for productive discussions.  Andrews gratefully acknowledges support from the Silverman (1968) Family Career Development Chair at MIT.
Armstrong gratefully acknowledges support from National Science Foundation Grant SES-1628939.}}
\maketitle
\vspace{-.25in}
\begin{abstract}
We derive mean-unbiased estimators for the structural parameter in instrumental variables models with a single endogenous regressor where the sign of one or more first stage coefficients is known.  In the case with a single instrument, there is a unique non-randomized unbiased estimator based on the reduced-form and first-stage regression estimates.  For cases with multiple instruments we propose a class of unbiased estimators and show that an estimator within this class is efficient when the instruments are strong.  We show numerically that unbiasedness does not come at a cost of increased dispersion in models with a single instrument: in this case the unbiased estimator is less dispersed than the 2SLS estimator.  Our finite-sample results apply to normal models with known variance for the reduced-form errors, and
imply analogous results under
weak instrument asymptotics with an unknown error distribution.
\end{abstract}

\section{Introduction}

Researchers often have strong prior beliefs about the sign of the first stage coefficient in instrumental variables models, to the point where the sign can reasonably be treated as known.
This paper shows that knowledge of the sign of the first stage coefficient allows us to construct an estimator for the coefficient on the endogenous regressor which is unbiased in finite samples when the reduced form errors are normal with known variance.
When the distribution of the reduced form errors is unknown, our results lead to estimators that are asymptotically unbiased under weak IV sequences as defined in \citet{staiger_instrumental_1997}.

As is well known, the conventional
two-stage least squares (2SLS) estimator may be severely biased in overidentified models with weak instruments.  Indeed the most common pretest for weak instruments, the %
\citet{staiger_instrumental_1997}
rule of thumb which declares the instruments weak when the first stage F statistic is less than 10, is shown in  \citet{StockYogo2005} to correspond to a test for the worst-case bias in 2SLS relative to OLS.  While the 2SLS estimator performs better in the just-identified case according to some measures of central tendency, in this case it has no first moment.\footnote{If we instead consider median bias, 2SLS exhibits median bias when the instruments are weak, though this bias decreases rapidly with the strength of the instruments.}  A number of papers have proposed alternative estimators to reduce particular measures of bias, e.g. \citet{AngristKrueger1995}, \citet{Imbensetal1999}, \citet{DonaldNewey2001}, \citet{AckerbergDevereux2009}, and \citet{Hardingetal2015}, but none of the resulting feasible estimators is unbiased either in finite samples or under weak instrument asymptotics.  Indeed,  \citet{HiranoPoter2015} show that mean, median, and quantile unbiased estimation are all impossible in the linear IV model with an unrestricted parameter space for the first stage.

We show that by exploiting information about the 
sign of the first stage we can circumvent this impossibility result and construct an unbiased estimator.
Moreover, the resulting estimators have a number of properties which make them appealing for applications.  In models with a single instrumental variable, which include many empirical applications, we show that there is a unique unbiased estimator based on the reduced-form and first-stage regression estimates.  Moreover, we show that this estimator is substantially
less dispersed that the usual 2SLS estimator in finite samples. Under standard (``strong instrument'') asymptotics, the unbiased estimator has the same asymptotic distribution as 2SLS, and so is asymptotically efficient in the usual sense.
In over-identified models many unbiased estimators exist, and we propose unbiased estimators which are asymptotically efficient when the instruments are strong.  Further, we show that in over-identified models we can construct unbiased estimators which are robust to small violations of the first stage sign restriction. We also derive a lower bound on the risk of unbiased estimators in finite samples, and show that this bound is attained in some models.

In contrast to much of the recent weak instruments literature, the focus of this paper is on estimation rather than hypothesis testing or confidence set construction.  Our approach is closely related to the classical theory of optimal point estimation (see e.g. \citet{LehmannCasella1998}) in that we seek estimators which perform well according to conventional estimation criteria (e.g. risk with respect to a convex loss function) within the class of unbiased estimators. As we note in Section \ref{sec: tests and CS} below it is straightforward to use results from the weak instruments literature to construct identification-robust tests and confidence sets based on our estimators.  As we also note in that section, however, optimal estimation and testing are distinct problems in models with weak instruments and it is not in general the case that optimal estimators correspond to optimal confidence sets or vice versa.  Given the important role played by both estimation and confidence set construction in empirical practice, our results therefore complement the literature on identification-robust testing.

The rest of this section discusses the assumption of known first stage sign, introduces the setting and notation, and briefly reviews the related literature.
Section \ref{single_instrument_sec} introduces the unbiased estimator for models with a single excluded instrument.
Section \ref{multi_instrument_sec} treats models with multiple instruments and introduces unbiased estimators which are robust to small violations of the first stage sign restriction.
Section \ref{sec: Simulation Results} presents simulation results on the performance of our unbiased estimators.
Section \ref{sec: empirical application} discusses illustrative applications using data from
\citet{Hornung2014} and \citet{angrist_does_1991}.
Proofs and auxiliary results are given in a separate appendix.\footnote{The appendix is available online at https://sites.google.com/site/isaiahandrews/working-papers}

\subsection{Knowledge of the First-Stage Sign}

The results in this paper rely on knowledge of the first stage sign.  This is reasonable in many economic contexts.
In their study of schooling and earnings, for instance, \citet{angrist_does_1991} note that compulsory schooling laws in the United States allow those born earlier in the year to drop out after completing fewer years of school than those born later in the year.
Arguing that quarter of birth can reasonably be excluded from a wage equation, they use this fact to motivate quarter of birth as an instrument for schooling.
In this context, a sign restriction on the first stage amounts to an assumption that the mechanism claimed by \citeauthor{angrist_does_1991} works in the expected direction: those born earlier in the year tend to drop out earlier.
More generally, empirical researchers often have some mechanism in mind for why a model is identified at all (i.e. why the first stage coefficient is nonzero) that leads to a known sign for the direction of this mechanism (i.e. the sign of the first stage coefficient).

In settings with heterogeneous treatment effects, a first stage monotonicity assumption is often used to interpret instrumental variables estimates (see \citealt{imbens_identification_1994}, \citealt{heckman_understanding_2006}).
In the language of \citet{imbens_identification_1994}, the monotonicity assumption requires that either the entire population affected by the treatment be composed of ``compliers,'' or that the entire population affected by the treatment be composed of ``defiers.''
Once this assumption is made, our assumption that the sign of the first stage coefficient is known amounts to assuming the researcher knows which of these possibilities (compliers or defiers) holds.
Indeed, in the examples where they argue that monotonicity is plausible (involving draft lottery numbers in one case and intention to treat in another), \citet{imbens_identification_1994}
argue that all individuals affected by the treatment are ``compliers'' for a certain definition of the instrument.

It is important to note, however, that knowledge of the first stage sign is not always a reasonable assumption, and thus that the results of this paper are not always applicable.  In settings where the instrumental variables are indicators for groups without a natural ordering, for instance, one typically does not have prior information about signs of the first stage coefficients.
To give one example, \citet{aizer_juvenile_2013}
use the fact that judges are randomly assigned to study the effects of prison sentences on recidivism.  In this setting, knowledge of the first stage sign would require knowing a priori which judges are more strict.

\subsection{Setting}\label{setting_sec}

For the remainder of the paper, we suppose that we observe a sample of $T$ observations $\left(Y_{t},X_{t},Z_{t}'\right)$,
$t=1,...,T$ where $Y_{t}$ is an outcome variable, $X_{t}$ is a
scalar endogenous regressor, and $Z_{t}$ is a $k\times1$ vector
of instruments. Let $Y$ and $X$ be  $T\times1$ vectors with
row $t$ equal to $Y_{t}$ and $X_{t}$ respectively, and let $Z$
be a $T\times k$ matrix with row $t$ equal to $Z_{t}'$. The usual
linear IV model, written in reduced-form, is

\begin{equation}
\begin{array}{c}
Y=Z\pi\beta+U\\
X=Z\pi+V
\end{array}.\label{eq:Linear IV Model}
\end{equation}
To derive finite-sample results, we treat the instruments $Z$ as fixed and assume that the errors
$\left(U,V\right)$ are jointly normal with mean zero and known variance-covariance
matrix $Var\left(\left(U',V'\right)'\right).$\footnote{Following the weak instruments literature we focus on models with homogeneous $\beta$, which rules out heterogeneous treatment effect models with multiple instruments.  In models with treatment effect heterogeneity and a single instrument, however, our results immediately imply an unbiased estimator of the local average treatment effect.  In models with multiple instruments, on the other hand, one can use our results to construct unbiased estimators for linear combinations of the local average treatment effects on different instruments. (Since the endogenous variable $X$ is typically a binary treatment in such models, this discussion applies primarily to asymptotic unbiasedness as considered in Appendix \ref{estimated_variance_sec} rather than the finite sample model where $X$ and $Y$ are jointly normal.)}
As is standard (see,
for example, D. \citet{AndrewsMoreiraStock2006}), in contexts with
additional exogenous regressors $W$ (for example an intercept), we
define $Y,$ $X,$ $Z$ as the residuals after projecting out these
exogenous regressors. If we denote the reduced-form and first-stage
regression coefficients by $\xi_{1}$ and $\xi_{2}$, respectively,
we can see that 
\begin{equation}
\left(\begin{array}{c}
\xi_{1}\\
\xi_{2}
\end{array}\right)=\left(\begin{array}{c}
\left(Z'Z\right)^{-1}Z'Y\\
\left(Z'Z\right)^{-1}Z'X
\end{array}\right)\sim N\left(\left(\begin{array}{c}
\pi\beta\\
\pi
\end{array}\right),\left(\begin{array}{cc}
\Sigma_{11} & \Sigma_{12}\\
\Sigma_{21} & \Sigma_{22}
\end{array}\right)\right)\label{eq: Sufficient Statistics Def}
\end{equation}
for 
\begin{align}\label{Sigma_def}
\Sigma=\left(\begin{array}{cc}
\Sigma_{11} & \Sigma_{12}\\
\Sigma_{21} & \Sigma_{22}
\end{array}\right)
=\left(I_{2}\otimes \left(Z'Z\right)^{-1}Z'\right)Var\left(\left(U',V'\right)'\right)\left(I_{2}\otimes\left(Z'Z\right)^{-1}Z'\right)'.%
\end{align}
We assume throughout that $\Sigma$ is positive definite.
Following the literature (e.g. \citealt{MoreiraMoreira2013}), we consider estimation based solely on $(\xi_1,\xi_2)$, which are sufficient for $\left(\pi,\beta\right)$ in the special case where the errors $(U_t,V_t)$ are iid over $t$. All uniqueness and efficiency statements therefore restrict attention to the class of procedures which depend on the data though only these statistics.  The conventional
 generalized method of moments (GMM) estimators belong to this class, so this restriction still allows efficient estimation under strong instruments.
We assume that the sign of each component $\pi_{i}$ of
$\pi$ is known, and in particular assume that the parameter space
for $\left(\pi,\beta\right)$ is 
\begin{equation}
\Theta=\left\{ \left(\pi,\beta\right):\pi\in\Pi\subseteq\left(0,\infty\right)^{k},\beta\in B\right\} \label{eq:Parameter Space}
\end{equation}
for some sets $\Pi$ and $B$.  Note that once we take the sign of $\pi_i$ to be known, assuming $\pi_i>0$ is without loss of generality
since this can always be ensured by redefining $Z$.

In this paper we focus on models with fixed instruments, normal errors,
and known error covariance, which allows us to obtain finite-sample
results. As usual, these finite-sample results will imply asymptotic
results under mild regularity conditions. Even in models with random
instruments, non-normal errors, serial correlation, heteroskedasticity,
clustering, or any combination of these, the reduced-form and
first stage estimators will be jointly asymptotically normal with
consistently estimable covariance matrix $\Sigma$ under mild regularity
conditions. Consequently, the finite-sample results we develop here
will imply asymptotic results under both weak and strong instrument
asymptotics, where we simply define $\left(\xi_{1},\xi_{2}\right)$
as above and replace $\Sigma$ by an estimator for the variance of $\xi$ to obtain
feasible statistics.
Appendix \ref{estimated_variance_sec} provides the details of these results.\footnote{The feasible analogs of the finite-sample unbiased estimators discussed here are asymptotically unbiased in general models in the sense of converging in distribution to random variables with mean $\beta$.  Note that this does not imply convergence of the mean of the feasible estimators to $\beta$, since convergence in distribution does not suffice for convergence of moments.
  Our estimator is thus asymptotically unbiased under weak and strong instruments in the same sense that LIML and just-identified 2SLS, which do not in general have finite-sample moments, are asymptotically unbiased under strong instruments.}
In the main text, we focus on what we view as the most novel component of the paper: finite-sample
mean-unbiased estimation of $\beta$ in the normal problem (\ref{eq: Sufficient Statistics Def}).

\subsection{Related Literature}

Our unbiased IV estimators build  on results for unbiased estimation of the inverse of a normal mean discussed in \citet{VoinovNikulin1993}.
More broadly, the literature has considered unbiased estimators in numerous other contexts, and we refer the reader to \citeauthor{VoinovNikulin1993} for details and references.   Recent work by \citet{MuellerWang2015} develops a numerical approach for approximating optimal nearly unbiased estimators in variety of nonstandard settings, though they do not consider the linear IV model.
To our knowledge the only other paper to treat finite sample mean-unbiased estimation in IV models is \citet{HiranoPoter2015}, who find that unbiased estimators do not exist when the parameter space is unrestricted.
In our setting, the sign restriction on the first-stage coefficient leads to a parameter space that violates the assumptions of \citet{HiranoPoter2015}, so that the negative results in that paper do not apply.\footnote{In particular, the sign restriction violates Assumption 2.4 of \citet{HiranoPoter2015}, and so renders the negative result in Theorem 2.5 of that paper inapplicable.  See Appendix \ref{hp_append} for details.}
The nonexistence of unbiased estimators has been noted in other nonstandard econometric contexts by \citet{hirano_impossibility_2012}.

The broader literature on the finite sample properties of IV estimators is huge:  see \citet{Phillips1983} and \citet{Hillier2006} for references.
While this literature does not study unbiased estimation in finite samples, there has been substantial research on higher order asymptotic bias properties: see the references given in the first section of the introduction, as well as \citet{HanhHausmanKuersteiner2004} and the references therein.

Our interest in finite sample results for a normal model with known reduced form variance is motivated by the weak IV literature, where this model arises asymptotically under weak IV sequences as in \citet{staiger_instrumental_1997} (see also Appendix \ref{estimated_variance_sec}).  In contrast to \citeauthor{staiger_instrumental_1997}, however, our results allow for heteroskedastic, clustered, or serially correlated errors as in \citet{Kleibergen2007}.
The primary focus of recent work on weak instruments has, however, been on inference rather than estimation.
See \citet{Andrews2014} for additional references.

Sign restrictions have been used in other settings in the econometrics literature, although the focus is often on inference or on using sign restrictions to improve population bounds, rather than estimation.
Recent examples include
\citet{moon_inference_2013} and several papers cited therein, which use sign restrictions to partially identify vector autoregression models.
Inference for sign restricted parameters has been treated by
D. \citet{andrews_testing_2001}
and
\citet{gourieroux_likelihood_1982}, among others.

\section{Unbiased Estimation with a Single Instrument}\label{single_instrument_sec}

To introduce our unbiased estimators, we first focus on the just-identified
model with a single instrument, $k=1$.  We show that unbiased estimation of $\beta$
in this context is linked to unbiased estimation of the inverse of a normal mean.  Using this fact we construct an unbiased estimator for $\beta$, show that it is unique, and 
discuss some of its finite-sample properties.  We note the key role played by  the first stage sign restriction, and show that our estimator is  equivalent to 2SLS (and thus efficient) when the instruments are strong.

 In the just-identified context $\xi_{1}$
and $\xi_{2}$ are scalars and we write
\[
\Sigma=\left(\begin{array}{cc}
\Sigma_{11} & \Sigma_{12}\\
\Sigma_{21} & \Sigma_{22}
\end{array}\right)=\left(\begin{array}{cc}
\sigma_{1}^{2} & \sigma_{12}\\
\sigma_{12} & \sigma_{2}^{2}
\end{array}\right).
\]
The problem of estimating $\beta$ therefore
reduces to that of estimating 
\begin{equation}
\beta=\frac{\pi\beta}{\pi}=\frac{E\left[\xi_{1}\right]}{E\left[\xi_{2}\right]}.\label{eq: beta in Just-ID case}
\end{equation}
The conventional IV estimate $\hat{\beta}_{2SLS}=\frac{\xi_{1}}{\xi_{2}}$
is the natural sample-analog of (\ref{eq: beta in Just-ID case}).
As is well-known, however, this estimator has no integer moments. This lack of unbiasedness reflects the fact that the expectation
of the ratio of two random variables is not in general equal to the
ratio of their expectations. 

The form of (\ref{eq: beta in Just-ID case}) nonetheless suggests
an approach to deriving an unbiased estimator. Suppose we can construct
an estimator $\hat{\tau}$ which (a) is unbiased for $1/\pi$
and (b) depends on the data only through $\xi_{2}$. If we then define
\begin{equation}
\hat{\delta}\left(\xi,\Sigma\right)=\left(\xi_{1}-\frac{\sigma_{12}}{\sigma_{2}^{2}}\xi_{2}\right),\label{eq:delta hat definition}
\end{equation}
we have that $E\left[\hat{\delta}\right]=\pi\beta-\frac{\sigma_{12}}{\sigma_{2}^{2}}\pi$,
and $\hat{\delta}$ is independent of $\hat{\tau}$.\footnote{Note that the orthogonalization used to construct $\hat\delta$ is similar to that used by \citet{Kleibergen2002}, \citet{Moreira2003}, and the subsequent weak-IV literature to construct identification-robust tests.} Thus, $E\left[\hat{\tau}\hat{\delta}\right]=E\left[\hat{\tau}\right]E\left[\hat{\delta}\right]=\beta-\frac{\sigma_{12}}{\sigma_{2}^{2}}$,
and $\hat{\tau}\hat{\delta}+\frac{\sigma_{12}}{\sigma_{2}^{2}}$ will
be an unbiased estimator of $\beta$. Thus, the problem of unbiased
estimation of $\beta$ reduces to that of unbiased estimation of the
inverse of a normal mean.

\subsection{Unbiased Estimation of the Inverse of a Normal Mean}

A result from \citet{VoinovNikulin1993} shows that unbiased estimation
of $1/\pi$ is possible if we assume its sign is known.
Let $\Phi$ and $\phi$ denote the standard normal cdf and pdf respectively.

\begin{lemma} \label{Lemma:Voinov Nikulin}Define 
\[
\hat{\tau}\left(\xi_{2},\sigma_{2}^{2}\right)=\frac{1}{\sigma_{2}}\frac{1-\Phi\left(\xi_{2}/\sigma_{2}\right)}{\phi\left(\xi_{2}/\sigma_{2}\right)}.
\]
For all $\pi>0$, $E_{\pi}\left[\hat{\tau}\left(\xi_{2},\sigma_{2}^{2}\right)\right]=\frac{1}{\pi}.$

\end{lemma}

The derivation of $\hat{\tau}\left(\xi_{2},\sigma_{2}^{2}\right)$
in \citet{VoinovNikulin1993} relies on the theory of bilateral Laplace
transforms, and offers little by way of intuition. Verifying unbiasedness
is a straightforward calculus exercise, however: for the interested
reader, we work through the necessary derivations in the proof of
Lemma \ref{Lemma:Voinov Nikulin}. 

From the formula for $\hat\tau$, we can see that this estimator has two properties which are arguably desirable for a restricted estimate of $1/\pi$.  First, it is positive by definition, thereby incorporating the restriction that $\pi>0$.  Second, in the case where positivity of $\pi$ is obvious from the data ($\xi_2$ is very large relative to its standard deviation), it is close to the natural plug-in estimator $1/\xi_2$.
The second property is an immediate consequence of a well-known approximation to the tail of the normal cdf, which is used extensively in the literature on extreme value limit theorems for normal sequences and processes
(see Equation 1.5.4 in \citealt{leadbetter_extremes_1983}, and the remainder of that book for applications).
We discuss this further in Section \ref{sub:Large pi Behavior, Single Instrument Case}.

\subsection{Unbiased Estimation of $\beta$}

Given an unbiased estimator of $1/\pi$ which depends only on $\xi_{2}$,
we can construct an unbiased estimator of $\beta$ as suggested above.
Moreover, this estimator is unique.

\begin{theorem} \label{Thm: Unbiased Estimator, Just-identified Case}
Define 
\[
\begin{array}{c}
\hat{\beta}_{U}\left(\xi,\Sigma\right)=\hat{\tau}\left(\xi_{2},\sigma_{2}^{2}\right)\hat{\delta}\left(\xi,\Sigma\right)+\frac{\sigma_{12}}{\sigma_{2}^{2}}\\
=\frac{1}{\sigma_{2}}\frac{1-\Phi\left(\xi_{2}/\sigma_{2}\right)}{\phi\left(\xi_{2}/\sigma_{2}\right)}\left(\xi_{1}-\frac{\sigma_{12}}{\sigma_{2}^{2}}\xi_{2}\right)+\frac{\sigma_{12}}{\sigma_{2}^{2}}.
\end{array}
\]
The estimator $\hat{\beta}_{U}\left(\xi,\Sigma\right)$ is unbiased
for $\beta$ provided $\pi>0$. 

Moreover, if the parameter space (\ref{eq:Parameter Space}) contains
an open set then $\hat{\beta}_{U}\left(\xi,\Sigma\right)$ is the
unique non-randomized unbiased estimator for $\beta$, in the sense
that any other estimator $\hat{\beta}\left(\xi,\Sigma\right)$ satisfying
\[
E_{\pi,\beta}\left[\hat{\beta}\left(\xi,\Sigma\right)\right]=\beta\,\,\forall\pi\in\Pi,\beta\in B
\]
also satisfies 
\[
\hat{\beta}\left(\xi,\Sigma\right)=\hat{\beta}_{U}\left(\xi,\Sigma\right)\, a.s.\,\forall\pi\in\Pi,\beta\in B.
\]

\end{theorem}

Note that the conventional IV estimator can be written as 
\[
\hat{\beta}_{2SLS}=\frac{\xi_{1}}{\xi_{2}}=\frac{1}{\xi_{2}}\left(\xi_{1}-\frac{\sigma_{12}}{\sigma_{2}^{2}}\xi_{2}\right)+\frac{\sigma_{12}}{\sigma_{2}^{2}}.
\]
Thus, $\hat{\beta}_{U}$ differs from the conventional IV estimator
only in that it replaces the plug-in estimate $1/\xi_{2}$ for $1/\pi$
by the unbiased estimate $\hat{\tau}$.  From results in e.g. \citet{baricz_mills_2008},
we have that $\hat{\tau}<1/\xi_{2}$ for $\xi_2>0$, so when $\xi_2$ is positive
$\hat\beta_U$ shrinks the conventional IV estimator towards $\sigma_{12}/\sigma_2^2$.\footnote{Under weak instrument asymptotics as in \citet{staiger_instrumental_1997} and homoskedastic errors, $\sigma_{12}/\sigma_2^2$ is the probability limit of the OLS estimator, though this does not in general hold under weaker assumptions on the error structure.}  By contrast, when 
$\xi_2<0$,  $\hat\beta_U$ lies on the opposite side of $\sigma_{12}/\sigma_2^2$  from the conventional IV estimator.  Interestingly,
one can show that the unbiased estimator is uniformly more likely to correctly sign $\beta-\frac{\sigma_{12}}{\sigma_2^2}$ than
is the conventional estimator, in the sense that for $\varphi(x)=1\{x\ge0\}$,
$$Pr_{\pi,\beta}\left\{\varphi\left(\hat\beta_U-\frac{\sigma_{12}}{\sigma_2^2}\right)=\varphi\left(\beta-\frac{\sigma_{12}}{\sigma_2^2}\right)\right\}\ge Pr_{\pi,\beta}\left\{\varphi\left(\hat\beta_{2SLS}-\frac{\sigma_{12}}{\sigma_2^2}\right)=\varphi\left(\beta-\frac{\sigma_{12}}{\sigma_2^2}\right)\right\},$$
 with strict inequality at some points.\footnote{This property is far from unique to the unbiased estimator, however.}

\subsection{Risk and Moments of the Unbiased Estimator}\label{single_instrument_risk_sec}

The uniqueness of $\hat\beta_U$ among nonrandomized estimators implies that $\hat\beta_U$ minimizes the risk $E_{\pi,\beta}\ell\left(\tilde\beta(\xi,\Sigma)-\beta\right)$ uniformly over $\pi,\beta$ and over the class of unbiased estimators $\tilde\beta$ for any loss function $\ell$ such that randomization cannot reduce risk.  In particular, by Jensen's inequality $\hat\beta_U$ is uniformly minimum risk for any convex loss function $\ell$.  This includes absolute value loss as well as squared error loss or $L^p$ loss for any $p\ge 1$.  However, elementary calculations show that $|\hat\beta_U|$ has an infinite $p$th moment for $p>1$.  Thus the fact that $\hat\beta_U$ has uniformly minimal risk implies that any unbiased estimator must have an infinite $p$th moment for any $p>1$.  In particular, while $\hat\beta_U$ is the uniform minimum mean absolute deviation unbiased estimator of $\beta$, it is minimum variance unbiased only in the sense that all unbiased estimators have infinite variance.  We record this result in the following theorem.

\begin{theorem}\label{betau_second_moment_thm}
For $\varepsilon>0$, the expectation of $|\hat\beta_U(\xi,\Sigma)|^{1+\varepsilon}$ is infinite for all $\pi,\beta$.  Moreover, if the parameter space (\ref{eq:Parameter Space}) contains
an open set then any unbiased estimator of $\beta$ has an infinite $1+\varepsilon$ moment.
\end{theorem}

\subsection{Relation to Tests and Confidence Sets}\label{sec: tests and CS}

As we show in the next subsection, $\hat\beta_U$ is asymptotically equivalent to 2SLS when the instruments are strong and so can be used together with conventional standard errors in that case.  Even when the instruments are weak the conditioning approach of \citet{Moreira2003} yields valid conditional critical values for arbitrary test statistics and so can be used to construct conditional t-tests based on $\hat\beta_U$ which control size.
We note, however, that
optimal estimation and optimal testing are distinct questions in the context of weak IV (e.g. while $\hat\beta_U$ is uniformly minimum risk unbiased for convex loss, it follows from the
results of \citet{Moreira2009} that the Anderson-Rubin test, rather than a conditional t-test based on $\hat\beta_U$, is the uniformly most powerful unbiased two-sided test in the present just-identified context).\footnote{\citet{Moreira2009} establishes this result in the model without a sign restriction, and it is straightforward to show that the result continues to hold in the sign-restricted model.}
Since our focus in this paper is on estimation we do not further pursue the question of optimal testing in this paper.
However, properties of tests based on unbiased estimators, particularly in contexts where the Anderson-Rubin test is not uniformly most powerful unbiased (such as one-sided testing and testing in the overidentified model of Section \ref{multi_instrument_sec}), is an interesting topic for future work.\footnote{Absent such results, we suggest reporting the Anderson-Rubin confidence set to accompany the unbiased point estimate.  As discussed in Section  \ref{beta_U Containment}, the 95\% Anderson-Rubin confidence set contains $\hat{\beta}_{U}$ with probability exceeding 97\%, and with probability near 100\% except when $\pi$ is extremely small.}

\subsection{Behavior of $\hat{\beta}_{U}$ When $\pi$ is Large\label{sub:Large pi Behavior, Single Instrument Case}}

While the finite-sample unbiasedness of $\hat{\beta}_{U}$ is appealing,
it is also natural to consider performance when the instruments
are highly informative. This situation, which we will model by taking
$\pi$ to be large, corresponds to the conventional strong-instrument
asymptotics where one fixes the data generating process and takes
the sample size to infinity.%
\footnote{Formally, in the finite-sample normal IV model (\ref{eq:Linear IV Model}),
strong-instrument asymptotics will correspond to fixing $\pi$ and
taking $T\to\infty$, which under mild conditions on $Z$ and $Var\left(\left(U',V'\right)'\right)$
will result in $\Sigma\to0$ in (\ref{eq: Sufficient Statistics Def}).
However, it is straightforward to show that the behavior of $\hat{\beta}_{U}$,
$\hat{\beta}_{2SLS}$, and many other estimators in this case will be
the same as the behavior obtained by holding $\Sigma$ fixed and taking
$\pi$ to infinity.
We focus on the latter case here to simplify the exposition.
See Appendix \ref{estimated_variance_sec}, which provides asymptotic results with an unknown error distribution, for asymptotic results under $T\to\infty$.
}

As we discussed above, the unbiased and conventional IV estimators
differ only in that the former substitutes $\hat{\tau}\left(\xi_{2},\sigma_{2}^{2}\right)$
for $1/\xi_{2}$. These two estimators
for $1/\pi$ coincide to a high order of approximation for large values
of $\xi_{2}$. Specifically, as noted in \citet{Small2010} (Section 2.3.4), for $\xi_{2}>0$ we have 
\[
\sigma_{2}\left|\hat{\tau}\left(\xi_{2},\sigma_{2}^{2}\right)-\frac{1}{\xi_{2}}\right|\le\left|\frac{\sigma_{2}^{3}}{\xi_{2}^{3}}\right|.
\]
Thus, since $\xi_{2}\stackrel{p}{\to}\infty$ as $\pi\to\infty$, the difference
between $\hat{\tau}\left(\xi_{2},\sigma_{2}^{2}\right)$ and $1/\xi_{2}$
converges rapidly to zero (in probability) as $\pi$ grows. Consequently,
the unbiased estimator $\hat{\beta}_{U}$ (appropriately normalized)
has the same limiting distribution as the conventional IV estimator
$\hat{\beta}_{2SLS}$ as we take $\pi\to\infty$.

\begin{theorem} \label{Thm: Limiting Distribution, Just-identified Case}
As $\pi\to\infty$, holding $\beta$ and $\Sigma$ fixed, 
\[
\pi\left(\hat{\beta}_{U}-\hat{\beta}_{2SLS}\right)\stackrel{p}{\to}0.
\]
Consequently, $\hat{\beta}_{U}\stackrel{p}{\to}\beta$ and 
\[
\pi\left(\hat{\beta}_{U}-\beta\right)\stackrel{d}{\to}N\left(0,\sigma_{1}^{2}-2\beta\sigma_{12}+\beta^{2}\sigma_{2}^{2}\right).
\]

\end{theorem}

Thus, the unbiased estimator $\hat{\beta}_{U}$ behaves as the standard IV estimator for large values of $\pi$.
Consequently, one can show that using this estimator along with conventional
standard errors will yield asymptotically valid inference under strong-instrument
asymptotics.
See Appendix \ref{estimated_variance_sec} for details.

\section{Unbiased Estimation with Multiple Instruments}\label{multi_instrument_sec}

We now consider the case with multiple instruments, where the model is given by (\ref{eq:Linear IV Model}) and (\ref{eq: Sufficient Statistics Def}) with $k$ (the dimension of $Z_t$, $\pi$, $\xi_1$ and $\xi_2$) greater than $1$.  As in Section \ref{setting_sec}, we assume that the sign of each element $\pi_i$ of the first stage vector is known, and we normalize this sign to be positive, giving the parameter space (\ref{eq:Parameter Space}).  

Using the results in Section \ref{single_instrument_sec} one can construct an unbiased estimator for $\beta$ in many different ways.  For any index $i\in\{1,\ldots,k\}$, the unbiased estimator based on $(\xi_{1,i},\xi_{2,i})$ will, of course, still be unbiased for $\beta$ when $k>1$.  One can also take non-random weighted averages of the unbiased estimators based on different instruments. Using the unbiased estimator based on a fixed linear combination of instruments  is another possibility, so long as the linear combination preserves the sign restriction.
However, such approaches will not adapt to information from the data about the relative strength of instruments and so will typically be inefficient when the instruments are strong.

By contrast, the usual 2SLS estimator achieves asymptotic efficiency in the strongly identified case (modeled here by taking $\|\pi\|\to\infty$) when errors are homoskedastic.  In fact, in this case 2SLS is asymptotically equivalent to an infeasible estimator that uses knowledge of $\pi$ to choose the optimal combination of instruments.
Thus, a reasonable goal is to construct an estimator that (1) is unbiased for fixed $\pi$ and (2) is asymptotically equivalent to 2SLS as $\|\pi\|\to\infty$.\footnote{In the heteroskedastic case, the 2SLS estimator will no longer be asymptotically efficient, and a two-step GMM estimator can be used to achieve the efficiency bound.  Because it leads to simpler exposition, and because the 2SLS estimator is common in practice, we consider asymptotic equivalence with 2SLS, rather than asymptotic efficiency in the heteroskedastic case, as our goal.  As discussed in Appendix \ref{multi_instrument_append}, however, our approach generalizes directly to efficient estimators in non-homoskedastic settings.}
In the remainder of this section we first introduce a class of unbiased estimators and then show that a (feasible) estimator in this class attains the desired strong IV efficiency property.  Further, we show that in the over-identified case it is possible to construct unbiased estimators which are robust to small violations of the first stage sign restriction.  Finally, we derive bounds on the attainable risk of any estimator for finite $\|\pi\|$ and show that, while the unbiased estimators described above achieve optimality in an asymptotic sense as $\|\pi\|\to\infty$ regardless of the direction of $\pi$, the optimal unbiased estimator for finite $\pi$ will depend on the direction of $\pi$.

\subsection{A Class of Unbiased Estimators}

Let
\begin{align*}
\xi(i)=\left(\begin{array}{c}
\xi_{1,i}  \\
\xi_{2,i}  \\
\end{array}\right)
\text{ and }
\Sigma(i)=\left(\begin{array}{cc}
\Sigma_{11,ii}  & \Sigma_{12,ii}  \\
\Sigma_{21,ii}  & \Sigma_{22,ii}  \\
\end{array}\right)
\end{align*}
be the reduced form and first stage coefficients on the $i$th instrument and their variance matrix, respectively, so that $\hat\beta_U(\xi(i),\Sigma(i))$ is the unbiased estimator based on the $i$th instrument.
Given a weight vector $w\in\mathbb{R}^k$ with $\sum_{i=1}^k w_i=1$, let
\begin{align*}
\hat\beta_w(\xi,\Sigma;w)
=\sum_{i=1}^k w_i\hat\beta_U(\xi(i),\Sigma(i)).
\end{align*}
Clearly, $\hat\beta_w$ is unbiased so long as $w$ is nonrandom.  Allowing $w$ to depend on the data $\xi$, however, may introduce bias through the dependence between the weights and the estimators $\hat\beta_U(\xi(i),\Sigma(i))$.

To avoid this bias we first consider a randomized unbiased estimator and then take its conditional expectation given the sufficient 
statistic $\xi$ to eliminate the randomization.
Let
$\zeta\sim N(0,\Sigma)$ be independent of $\xi$, and let
$\xi^{(a)}=\xi+\zeta$ and $\xi^{(b)}=\xi-\zeta$.  Then $\xi^{(a)}$ and $\xi^{(b)}$ are (unconditionally) independent draws with the same marginal distribution as $\xi$, save that $\Sigma$ is replaced by $2\Sigma$.  If $T$ is even,  $Z'Z$ is the same across the first and second halves of the sample,
and the errors are iid, then $\xi^{(a)}$ and $\xi^{(b)}$ have the same joint distribution as the reduced form estimators based on the first and second half of the sample. Thus, we can think of these as split-sample reduced-form estimates.

Let $\hat w=\hat w(\xi^{(b)})$ be a vector of data dependent weights with $\sum_{i=1}^k \hat w_i=1$.  By the independence of $\xi^{(a)}$ and $\xi^{(b)}$,
\begin{align}\label{hatbetaw_unbiasedness_eq}
E\left[\hat\beta_w(\xi^{(a)},2\Sigma;\hat w(\xi^{(b)}))\right]
=\sum_{i=1}^k E\left[\hat w_i(\xi^{(b)})\right]
\cdot E\left[\hat\beta_U(\xi^{(a)}(i),2\Sigma(i))\right]
=\beta.
\end{align}
To eliminate the noise introduced by $\zeta$,
define the ``Rao-Blackwellized'' estimator
\begin{align*}
\hat\beta_{RB}=\hat\beta_{RB}(\xi,\Sigma;\hat w)
=E\left[\hat\beta_w(\xi^{(a)},2\Sigma;\hat w(\xi^{(b)}))\Big|\xi\right].
\end{align*}
This gives a class of unbiased estimators, where the estimator depends on the choice of the weight $\hat w$.
Unbiasedness of $\hat\beta_{RB}$ follows immediately from (\ref{hatbetaw_unbiasedness_eq}) and the law of iterated expectations. 
While $\hat\beta_{RB}$ does not, to our knowledge, have a simple closed form, it can be computed by integrating over the distribution of $\zeta$.  This can easily be  done by simulation, taking the sample average of $\hat\beta_w$ over simulated draws of $\xi^{(a)}$ and $\xi^{(b)}$ while holding $\xi$ at its observed value.

\subsection{Equivalence with 2SLS under Strong IV Asymptotics}\label{sec: Homoskedastic Efficiency}

We now propose a set of weights $\hat w$ which yield an unbiased estimator asymptotically equivalent to 2SLS.  To motivate these weights, note that for $W=Z'Z$ and $e_i$ the $i$th standard basis vector, the 2SLS estimator can be written as
\begin{align*}
\hat\beta_{2SLS}=\frac{\xi_2' W \xi_1}{\xi_2' W \xi_2}
=\sum_{i=1}^k \frac{\xi_2' W e_ie_i' \xi_2}{\xi_2' W \xi_2}\frac{\xi_{1,i}}{\xi_{2,i}},
\end{align*}
which is the GMM estimator with weight matrix $W=Z'Z$.
Thus, the 2SLS estimator is a weighted average of the 2SLS estimates based on single instruments, where the
 weight for estimate $\xi_{1,i}/\xi_{2,i}$ based on instrument $i$ is equal to
 $\frac{\xi_2' W e_ie_i' \xi_2}{\xi_2' W \xi_2}$.  This suggests the unbiased Rao-Blackwellized estimator with weights
$\hat w_i^*(\xi^{(b)})=\frac{{\xi_2^{(b)}}' W e_ie_i' {\xi_2^{(b)}}}{{\xi_2^{(b)}}' W {\xi_2^{(b)}}}$:
\begin{align}\label{iv_weight_rb_estimator_eq}
\hat\beta_{RB}^*
=\hat\beta_{RB}(\xi,\Sigma;\hat w)
=E\left[\hat\beta_w(\xi^{(a)},2\Sigma;\hat w^*(\xi^{(b)}))\Big|\xi\right].
\end{align}

The following theorem shows that $\hat\beta_{RB}^*$ is asymptotically equivalent to $\hat\beta_{2SLS}$ in the strongly identified case, and is therefore asymptotically efficient if the errors are iid.

\begin{theorem}\label{rb_asym_efficiency_thm}
Let $\|\pi\|\to\infty$ with $\|\pi\|/\min_i\pi_i=\mathcal{O}(1)$.  Then
$\|\pi\|(\hat\beta_{RB}^*-\hat\beta_{2SLS})\stackrel{p}{\to} 0$.
\end{theorem}
The condition that $\|\pi\|/\min_i\pi_i=\mathcal{O}(1)$ amounts to an assumption that the ``strength'' of all instruments is of the same order.  As discussed below in Section \ref{robustness_sec}, this assumption can be relaxed by redefining the instruments.

To understand why Theorem \ref{rb_asym_efficiency_thm} holds, consider the ``oracle'' weights $w_i^*=\frac{\pi'W e_ie_i'\pi}{\pi'W\pi}$.  It is easy to see that $\hat w_i^*-w_i^*\stackrel{p}{\to} 0$ as $\|\pi\|\to \infty$.  Consider the oracle unbiased estimator
$\hat\beta_{RB}^o=\hat\beta_{RB}(\xi,\Sigma;w^*)$, and the oracle combination of individual 2SLS estimators
$\hat\beta_{2SLS}^o=\sum_{i=1}^kw_i^*\frac{\xi_{1,i}}{\xi_{2,i}}$.
By arguments similar to those used to show that statistical noise in the first stage estimates does not affect the 2SLS asymptotic distribution under strong instrument asymptotics, it can be seen that $\|\pi\|(\hat\beta_{2SLS}^o-\hat\beta_{2SLS})\stackrel{p}{\to} 0$ as $\|\pi\|\to\infty$.
Further, one can show that
$\hat\beta_{RB}^o=\hat\beta_{w}(\xi,\Sigma;w^*)=\sum_{i=1}^k w_i^*\hat\beta_U(\xi(i),\Sigma(i))$.
Since this is just $\hat\beta_{2SLS}^o$ with $\hat\beta_U(\xi(i),\Sigma(i))$ replacing $\xi_{i,1}/\xi_{i,2}$, it follows
by Theorem \ref{Thm: Limiting Distribution, Just-identified Case}
that
$\|\pi\|(\hat\beta_{RB}^o-\hat\beta_{2SLS}^o)\stackrel{p}{\to} 0$.
Theorem \ref{rb_asym_efficiency_thm} then follows by showing that
$\|\pi\|(\hat\beta_{RB}-\hat\beta_{RB}^o)\stackrel{p}{\to} 0$,
which follows for essentially the same reasons that first stage noise does not affect the asymptotic distribution of the 2SLS estimator but requires some additional argument.  We refer the interested reader to the proof of Theorem \ref{rb_asym_efficiency_thm} in Appendix \ref{proof_append} for details.

\subsection{Robust Unbiased Estimation}\label{robustness_sec}

So far, all the unbiased estimators we have discussed required $\pi_i>0$ for all $i$.  Even when the first stage sign is dictated by theory, however, we may be concerned that this restriction may fail to hold exactly in a given empirical context.  To address such concerns, in this section we show that in over-identified 
models we can construct estimators which are robust to small violations of the sign restriction.  Our approach has the further benefit of ensuring asymptotic efficiency when, while $\|\pi\|\to\infty$, the elements $\pi_i$ may increase at different rates.

Let $M$ be a $k\times k$ invertible matrix such that all elements are strictly positive, and
\begin{align*}
\tilde\xi=(I_2\otimes M)\xi,&
&\tilde\Sigma=(I_2\otimes M)\Sigma(I_2\otimes M)',&
&\tilde W={M^{-1}}'W M^{-1}.
\end{align*}
The GMM estimator based on $\tilde\xi$ and $\tilde W$ is numerically equivalent to the GMM estimator based on $\xi$ and $W$.  In particular, for many choices of $W$, including all those discussed above, estimation based on $(\tilde\xi,\tilde W,\tilde\Sigma)$  is equivalent to estimation based on instruments $ZM^{-1}$ rather than $Z$.  

Note that for $\tilde\pi=M\pi$, $\tilde\xi$ is normally distributed with mean $(\tilde\pi'\beta,\tilde\pi')'$ and variance $\tilde \Sigma$.  
Thus, if we construct the estimator $\hat\beta_{RB}^*$ from $(\tilde\xi,\tilde W,\tilde\Sigma)$ instead of $(\xi,W,\Sigma),$ we obtain an unbiased estimator provided $\tilde\pi_i>0$ for all $i$. Since all elements of $M$ are strictly positive this is a strictly weaker condition than $\pi_i>0$ for all $i$.
By Theorem \ref{rb_asym_efficiency_thm}, $\hat\beta_{RB}^*$ constructed from  from $\tilde\xi$ and $\tilde W$ will be asymptotically efficient as $\|\tilde\pi\|\to\infty$ so long as $\tilde\pi=M \pi$ is nonnegative and satisfies $\|\tilde\pi\|/\min_i\tilde\pi_i=\mathcal{O}(1)$.
Note, however, that
$$\min_i\tilde\pi_i\ge (\min_{i,j} M_{ij})\|\pi\|
=(\min_{i,j} M_{ij})\frac{\|\pi\|}{\|M \pi\|} \|\tilde\pi\|
\ge (\min_{i,j} M_{ij})\left(\inf_{\|u\|=1}\frac{\|u\|}{\|M u\|}\right) \|\tilde\pi\|$$
so  $\|\tilde\pi\|/\min_i\tilde\pi_i=\mathcal{O}(1)$ now follows automatically from $\|\pi\|\to\infty$.

Conducting estimation based on $\tilde\xi$ and $\tilde W$ offers a number of advantages for many different choices of $M$.
One natural class of transformations $M$ is 
\begin{align}\label{eq: class of M}
M=\left[\begin{array}{ccccc}
1 & c & c & \cdots & c\\
c & 1 & c & \cdots & c\\
c & c & 1 & \cdots & c\\
\vdots & \vdots & \vdots & \ddots & \vdots\\
c & c & c & \cdots & 1
\end{array}\right] Diag(\Sigma_{22})^{-\frac{1}{2}}
\end{align}
for $c\in[0,1)$ and  $Diag(\Sigma_{22})$ the matrix with the same diagonal as $\Sigma_{22}$ and zeros elsewhere. For a given $c$, denote the estimator $\hat\beta_{RB}^*$ based on the corresponding $(\tilde\xi,\tilde W,\tilde\Sigma)$ by $\hat\beta_{RB,c}^*$.  One can show that $\hat\beta_{RB,0}^*=\hat\beta_{RB}^*$ based on $(\xi,W,\Sigma)$, and going forward we let $\hat\beta_{RB}^*$ denote $\hat\beta_{RB,0}^*$.

We can interpret $c$ as specifying a level of robustness to violations on the sign restriction for $\pi_i$.  In particular, for a given choice of $c$, $\tilde\pi$ will satisfy the sign restriction provided that for each $i$, 
\begin{align*}
-\pi_i/\sqrt{\Sigma_{22,ii}}<c\cdot\sum_{j\neq i}\pi_j/\sqrt{\Sigma_{22,jj}},
\end{align*}
that is, provided the expected z-statistic for testing that each wrong-signed $\pi_i$ is equal to zero is less than $c$ times the sum of the expected
z-statistics for $j\neq i$.  Larger values of $c$ provide a greater degree of robustness
to violations of the sign restriction, while 
all choices of $c\in(0,1)$ yield asymptotically equivalent estimators as $\|\pi\|\to\infty$.  For finite values of
$\pi$ however, different choices of $c$ yield different estimators, so we explore the effects of different choices below using the \citet{angrist_does_1991} dataset.  Determining the optimal choice of $c$ for finite values of $\pi$ is an interesting topic for future research.

\subsection{Bounds on the Attainable Risk}\label{risk_lower_bound_subsec}

While the class of estimators given above has the desirable property of asymptotic efficiency as $\|\pi\|\to\infty$, it is useful to have a benchmark for the performance for finite $\pi$.  In Appendix \ref{risk_lower_bound_sec}, we derive a lower bound for the risk of any unbiased estimator at a given $\pi^*,\beta^*$.  The bound is based on the risk in a submodel with a single instrument and, as in the single instrument case, shows that any unbiased estimator must have an infinite $1+\varepsilon$ absolute moment for $\varepsilon>0$.  In certain cases, which include large parts of the parameter space under homoskedastic errors $(U_t,V_t)$, the bound can be attained.  The estimator that attains the bound turns out to depend on the value $\pi^*$, which shows that no uniform minimum risk unbiased estimator exists. See Appendix \ref{risk_lower_bound_sec} for details.

\section{Simulations}
\label{sec: Simulation Results}
In this section we present simulation results on the performance of our unbiased estimators.
We study the model with normal errors and known reduced-form variance.
We first consider models with a single instrument and then turn to over-identified models.
Since the parameter space in the single-instrument model is small, we are able to obtain comprehensive
simulation results in this case, studying performance over a wide range of parameter values.
In the over-identified case, by contrast, the parameter space is too large to comprehensively explore by simulation so we instead
calibrate our simulations to the \citet{staiger_instrumental_1997} specifications for the \citet{angrist_does_1991} dataset.

\subsection{Performance with a Single Instrument}

The estimator $\hat{\beta}_{U}$ based on a single instrument plays
a central role in all of our results, so in this section we examine
the performance of this estimator in simulation. For purposes of comparison we also
discuss results for the two-stage least squares estimator $\hat{\beta}_{2SLS}.$
The lack of moments for $\hat\beta_{2SLS}$ in the just-identified
context renders some comparisons with $\hat{\beta}_{U}$ infeasible, however,
so we also consider the performance of the \citet{Fuller1977} estimator
with constant one, 
\[
\hat{\beta}_{FULL}=\frac{\xi_{2}\xi_{1}+\sigma_{12}}{\xi_{2}^{2}+\sigma_{2}^{2}}
\]
which we define as in \citet{MoreiraMillsVilela2014}.%
\footnote{In the case where $U_{t}$ and $V_{t}$ are correlated or heteroskedastic
across $t$, the definition of $\hat{\beta}_{FULL}$ above is the natural
extension of the definition considered in \citet{MoreiraMillsVilela2014}.%
} Note that in the just-identified case considered here $\hat{\beta}_{FULL}$
also coincides with the bias-corrected 2SLS estimator (again, see \citeauthor{MoreiraMillsVilela2014}).

While the model (\ref{eq: Sufficient Statistics Def}) has five parameters in the single-instrument
case, $\left(\beta,\pi,\sigma_{1}^{2},\sigma_{12},\sigma_{2}^{2}\right)$,
an equivariance argument implies that for our purposes it suffices
to fix $\beta=0$, $\sigma_{1}=\sigma_{2}=1$ and consider the parameter
space $\left(\pi,\sigma_{12}\right)\in(0,\infty)\times[0,1)$. See
Appendix \ref{equivariance_append} for details. Since this parameter space is just two-dimensional, we can fully explore it via simulation.

\subsubsection{Estimator Location}

We first compare the bias of $\hat{\beta}_{U}$ and $\hat{\beta}_{FULL}$
(we omit $\hat{\beta}_{2SLS}$ from this comparison, as it does not
have a mean in the just-identified case). We consider $\sigma_{12}\in\left\{ 0.1,0.5,0.95\right\} $
and examine a wide range of values for $\pi>0$.%
\footnote{We restrict attention to $\pi\ge 0.16$ in the bias plots.  
Since the first stage F-statistic is $F=\xi_2^2$ in the present context,
this corresponds to $E[F]\ge 1.026$.
The expectation of $\hat{\beta}_{U}$ ceases to exist at $\pi=0$,
and for $\pi$ close to zero the heavy tails of $\hat{\beta}_{U}$
make computing the expectation very difficult.  Indeed, we use numerical integration rather than
monte-carlo integration here because it allows us to consider smaller values $\pi$.  We thank an anonymous referee for this
suggestion.} These results are plotted in the first panel of Figure \ref{fig: univariate results}.

If rather than mean bias we instead consider median
bias, we find that $\hat{\beta}_{U}$ and $\hat{\beta}_{2SLS}$ generally
exhibit smaller median bias than $\hat{\beta}_{FULL}$. There is no
ordering between $\hat{\beta}_{U}$ and $\hat{\beta}_{2SLS}$ in terms
of median bias, however, as the median bias of $\hat{\beta}_{U}$
is smaller than that of $\hat{\beta}_{2SLS}$ for very small values
of $\pi$, while the median bias of $\hat{\beta}_{2SLS}$ is smaller
for larger values $\pi$.  A plot of median bias is given in Appendix \ref{median bias append}.

\subsubsection{Estimator Absolute Deviation}\label{sec: absolute deviation sims}

We examine the distribution of the absolute deviation of each estimator from the true parameter value.  The last three panels of Figure \ref{fig: univariate results} plot the 10th, 50th, and 90th percentiles of absolute deviation of the estimators considered from the true value $\beta$ for three values of $\sigma_{12}$.  We plot the log quantiles of absolute deviation (or equivalently the quantiles of log absolute deviation) for the sake of visibility.
Here, and in additional unreported simulation results, we find that $\hat\beta_U$ has smaller median absolute deviation than $\hat\beta_{IV}$ uniformly over the parameter space.  The $10$th and $90$th percentiles of the absolute deviation are also lower for $\hat\beta_{U}$ than $\hat\beta_{IV}$ for much of the parameter space, though we find that there is not a uniform ranking for all percentiles.
The Fuller estimator has low median absolute deviation over much of the parameter space, but performs worse than both $\hat\beta_{U}$ and $\hat\beta_{IV}$ in certain cases, such as when $\sigma_{12}=0.95$ and the first stage coefficient is small.
Turning to mean absolute deviation, we find that the mean absolute deviation of $\hat\beta_U$ from $\beta$ exceeds that of $\hat\beta_{FULL}$ except in cases with very high $\rho$ and small $\pi$, while as already noted the mean absolute deviation of $\hat\beta_{IV}$ is infinite.

Thus, over much of the parameter space the unbiased estimator is more concentrated around the true parameter value than the 2SLS estimator, according to a variety of different measures of concentration.  It would be interesting to decompose the deviations from the true parameter value into bias and variance components.  Unfortunately, however, the lack of second moments of both the 2SLS and unbiased estimators means that the variance is infinite in both cases and therefore does not yield a useful comparison.  To get around this, we consider the distribution of the absolute deviation of each estimator from the median of the estimator as a location free measure of dispersion.  In Appendix \ref{dispersion_simulation_append}, we examine this numerically and find a stochastic dominance relation in which the unbiased estimator is less dispersed than the 2SLS estimator and more dispersed than the Fuller estimator uniformly over the parameter space.

\begin{figure}
\includegraphics[scale=0.65]{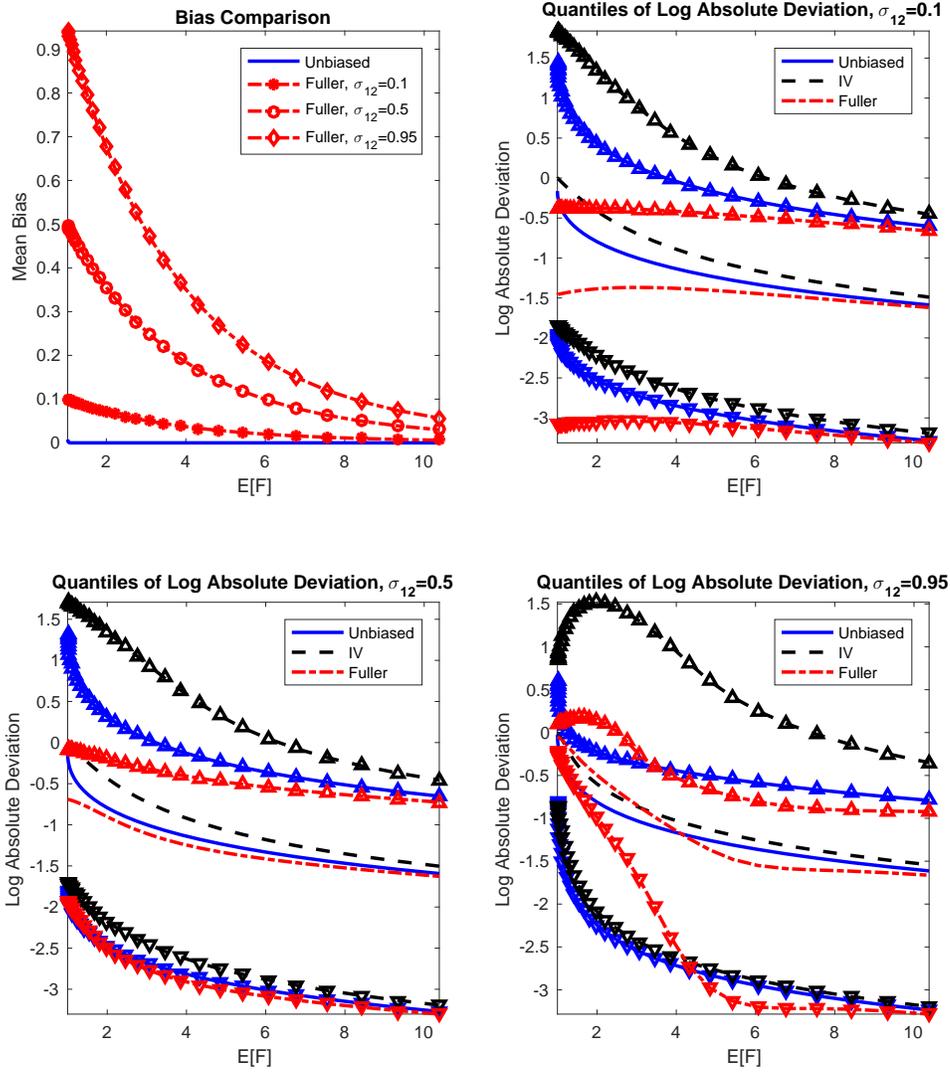}\protect\caption{The first panel plots the bias of single-instrument estimators, calculated by numerical integration, against the mean $E\left[F\right]$ of first-stage F-statistic. The remaining panels plot log quantiles of absolute deviation from the true value of $\beta$ for unbiased estimator, 2SLS, and Fuller, for three values of $\sigma_{12}$.  The lines corresponding to the median are plotted without markers, while the lines corresponding to the 90th and 10th percentiles are plotted with upward and downward pointing triangles, respectively.  The absolute deviation results are based on 10 million simulation draws.\label{fig: univariate results}}
\end{figure}

\subsection{Performance with Multiple Instruments}\label{sec: multi-instrument sims}

In models with multiple instruments, if we assume that
errors are homoskedastic an equivariance argument closely related
to that in just-identified case again allows us to reduce the dimension
of the parameter space.  Unlike in the just-identified case, however, the matrix $Z'Z$
and the direction of the first stage, $\pi/\|\pi\|$, continue
to matter (see Appendix \ref{equivariance_append} for details).  As a result, the parameter space is
too large to fully explore by simulation, so we instead calibrate our simulations to the \citet{staiger_instrumental_1997}
specifications for the 1930-1939
cohort in the \citet{angrist_does_1991} data. While there is statistically significant heteroskedasticity
in this data, this significance appears to be
the result of the large sample size rather than substantively important
deviations from homoskedasticity. In particular, procedures which assume homoskedasticity
produce very similar answers to heteroskedasticity-robust procedures when applied to this data.
Thus, given that homoskedasticity leads to a reduction of the parameter space as discussed above, we impose homoskedasticity in our simulations.

In each of the four \citet{staiger_instrumental_1997} specifications
we estimate $\pi/\|\pi\|$ and $Z'Z$ from the
data (ensuring, as discussed in  Appendix \ref{multi-inst sim design}, that $\pi/\|\pi\|$ satisfies the sign restriction).
After reducing the parameter space by equivariance and calibrating $Z'Z$ and $\pi/\|\pi\|$ to the data, the model
has two remaining free parameters: the norm of the first stage, $\|\pi\|$, and  the
correlation $\sigma_{UV}$ between the reduced-form and first-stage errors. We examine
behavior for a range of values for $\|\pi\|$ and for $\sigma_{UV}\in\left\{ 0.1,0.5,0.95\right\} .$
Further details on the simulation design are given in Appendix \ref{multi-inst sim design}.

For each parameter value we simulate the performance of $\hat{\beta}_{2SLS}$,
$\hat{\beta}_{FULL}$ (which is again the Fuller estimator
with constant equal to one), and $\hat{\beta}_{RB}^{*}$ as defined in Section \ref{sec: Homoskedastic Efficiency}.
We also consider the robust estimators $\hat\beta_{RB,c}^*$ discussed in Section \ref{robustness_sec}
for $c\in\left\{ 0.1,0.5,0.9\right\} $, but find that all three choices
produce very similar results and so focus on $c=0.5$ to simplify the graphs.\footnote{All results for the $RB$ estimators are based on $1,000$ draws of $\zeta$.} Even with a million simulation replications, simulation
estimates of the bias for the unbiased estimators (which we know to be zero from the results of Section \ref{multi_instrument_sec}) remain noisy relative
to e.g. the bias in 2SLS in some calibrations, so we do not plot the bias
estimates and instead focus on the mean absolute deviation (MAD) $E_{\pi,\beta}\left[\left|\hat{\beta}-\beta\right|\right]$
since, unlike in the just-identified case, the MAD
for 2SLS is now finite. We also plot the lower bound on the mean absolute
deviation of unbiased estimators discussed in Section \ref{risk_lower_bound_subsec}.
The results are plotting in Figure \ref{fig: Multivariate results}.

\begin{figure}
\includegraphics[scale=.65]{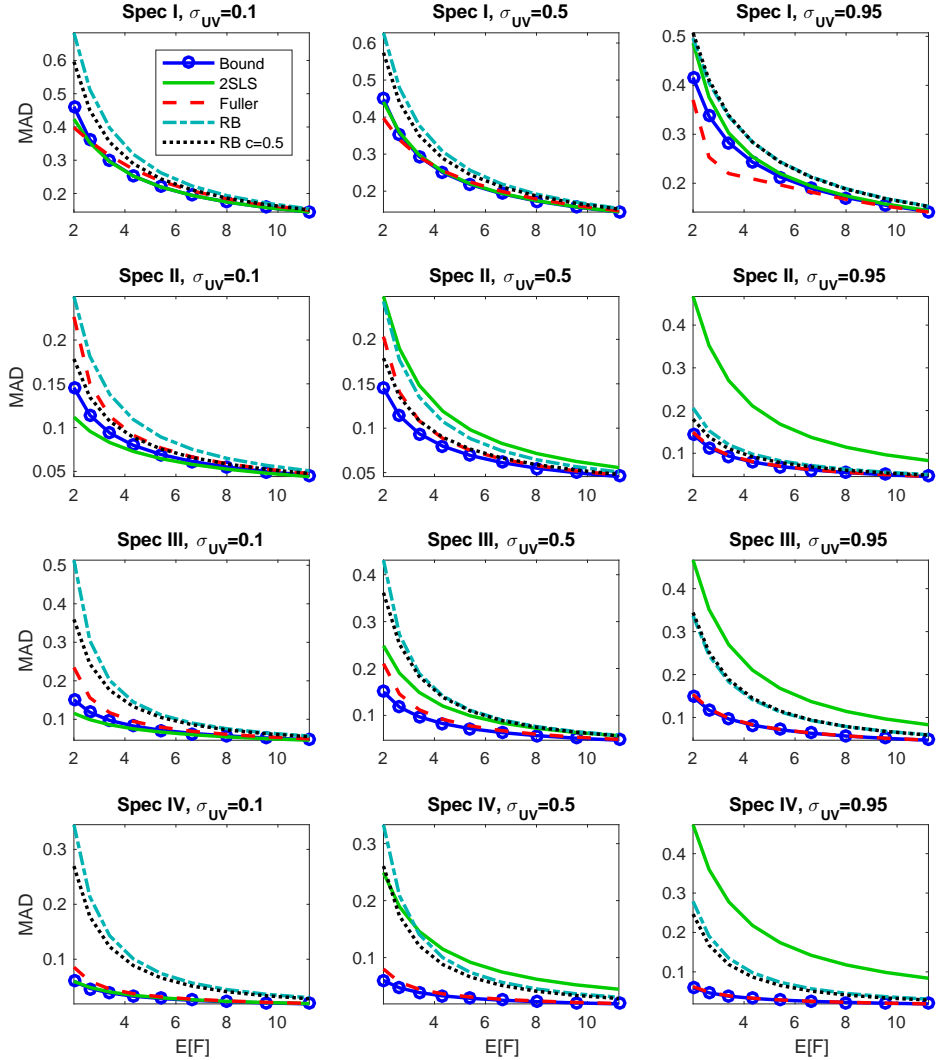}

\protect\caption{Mean absolute deviation of estimators in simulations calibrated to
specification I-IV of \citet{staiger_instrumental_1997}.  These specifications have $k=3,30,28,$ and $178$ instruments, respectively.  Results for specifications I-III are based on 1 million simulation draws, while results for specification IV are based on 100,000 simulation draws.\label{fig: Multivariate results}}
\end{figure}

Several features become clear from these results.  As expected, the performance of 2SLS is typically worse for models with more instruments or with a higher degree of correlation between the reduced-form and first-stage errors (i.e. higher $\sigma_{UV}$).  The robust unbiased estimator $\hat\beta_{RB,0.5}$
generally outperforms $\hat\beta_{RB}^*=\hat\beta_{RB,0}^*$.  Since the estimators with $c=0.1$ and $c=0.9$ perform very similarly to that with $c=0.5$, they outperform $\hat\beta_{RB}^*$ as well.  The gap in performance between the RB estimators and the lower bound on MAD over the class of all unbiased estimators is typically larger in specifications with more instruments.  Interestingly, we see that the Fuller estimator often performs quite well, and has MAD close to or below the lower bound for the class of unbiased estimators in most designs.  While this estimator is biased, its bias decreases quickly in $\|\pi\|$ in the designs considered.  Thus, at least in the homoskedastic case, this estimator seems a potentially appealing choice if we are willing to accept bias for small values of $\pi$.

\section{Empirical Applications}\label{sec: empirical application}

We calculate our proposed estimators in two empirical applications. First,
we consider the data and specifications used in \citet{Hornung2014} to
examine the effect of seventeenth century migrations on productivity. For our second application,
we study the \citet{staiger_instrumental_1997} specifications for the \citet{angrist_does_1991} dataset on
the relationship between education and labor market earnings.
Before continuing, we present a step-by-step description of the implementation of our estimators.

\subsection{Implementation}

To describe the implementation in a general setup,
we introduce additional notation
to explicitly allow for additional exogenous variables (such as a constant).  We have observations $t=1,\ldots,T$ with $\tilde Y_t$ a scalar outcome variable, $\tilde X_t$ a scalar endogenous variable, $\tilde Z_t$ a $k\times 1$ vector of instruments and $W_t$ a vector of additional control variables.
Let $\tilde Y=(\tilde Y_1,\ldots,\tilde Y_T)'$, $\tilde X=(\tilde X_1,\ldots,\tilde X_T)'$, $\tilde Z=(\tilde Z_1,\ldots,\tilde Z_T)'$ and $W=(W_1,\ldots,W_T)'$.
Let $Y=(I-W(W'W)^{-1}W')\tilde Y$, $X=(I-W(W'W)^{-1}W')\tilde X$ and $Z=(I-W(W'W)^{-1}W')\tilde Z$ denote the residuals from regressing $\tilde Y$, $\tilde X$ and $\tilde Z$ on $W$, as described in the introduction.

Our estimates are obtained using the following steps.
\begin{itemize}
\item[1.)] Let $\xi_1$ and $\xi_2$ denote the estimates of the coefficient on $\tilde Z_t$ in the regressions of $\tilde Y_t$ and $\tilde X_t$ respectively on $\tilde Z_t$ and $W_t$, and let $\hat U_t$ and $\hat V_t$ denote residuals from these regressions.  Let $\hat \Sigma$ denote an estimate of the variance-covariance matrix of $(\xi_1',\xi_2')'$.  If the observations are independent (but possibly heteroskedastic), we can use the heteroskedasticity robust estimate
\begin{align*}
(I_2\otimes (Z'Z)^{-1})
\left[\sum_{t=1}^T\left(\begin{array}{cc}
\hat U_t^2Z_tZ_t' & \hat U_t\hat V_tZ_tZ_t'  \\
\hat U_t\hat V_tZ_tZ_t' & \hat V_t^2Z_tZ_t'
\end{array}\right)
\right]
(I_2\otimes (Z'Z)^{-1}).
\end{align*}
We use this estimate in our application based on \citet{angrist_does_1991}, while for our application based on \citet{Hornung2014} we follow Hornung and use a clustering-robust variance estimator.  Likewise, in time-series contexts one could use a serial-correlation robust variance estimator, e.g. \citet{NeweyWest1987} here.

\item[2.)] In the case of a single instrument (so $Z_t$ is scalar), the estimate is given by $\hat\beta_U(\xi,\hat\Sigma)$ where $\hat\beta_U(\cdot,\cdot)$ is defined in Theorem \ref{Thm: Unbiased Estimator, Just-identified Case}.

\item[3.)] In the case with $k>1$ instruments,
let $\hat\Sigma_{22}$ denote the lower-right $k\times k$ submatrix of $\hat\Sigma$, and
let $M$ be the matrix given in (\ref{eq: class of M}) with $\Sigma_{22}$ replaced by $\hat\Sigma_{22}$ for some choice of $c$ between $0$ and $1$ (we find that $c=.5$ works well in our Monte Carlos).  Let $\tilde \xi=(I_2\otimes M)\xi$ and $\tilde \Sigma=(I_2\otimes M)\hat \Sigma (I_2\otimes M)'$.
Let $\tilde \Sigma(i)$ denote the $2\times 2$ symmetric matrix with diagonal elements given by the $i,i$ and $(k+i),(k+i)$ elements of $\tilde\Sigma$ respectively and off-diagonal element given by the $i,(k+i)$ element of $\tilde\Sigma$.
Generate $S$ independent $N(0,\tilde\Sigma)$ vectors $\zeta_1,\ldots,\zeta_S$.  Let
$\tilde\xi_1$ and $\zeta_{s,1}$ denote the $k\times 1$ vectors with elements $1$ through $k$ of $\tilde\xi$ and $\zeta_s$ respectively and let
$\tilde\xi_2$ and $\zeta_{s,2}$ denote the $k\times 1$ vectors with elements $k+1$ through $2k$ of $\tilde\xi$ and $\zeta_s$ respectively.
Let $\tilde\xi(i)=(\tilde\xi_{1,i},\tilde\xi_{2,i})'$ and let
$\tilde\zeta_s(i)=(\tilde\zeta_{s,1,i},\tilde\zeta_{s,2,i})'$.
Let
\begin{align*}
\hat\beta_{s}=\sum_{i=1}^k w_{i,s}\hat\beta_U(\tilde\xi(i)+\zeta_s(i),2\tilde\Sigma(i))
\end{align*}
where $\hat\beta_U(\cdot,\cdot)$ is defined in Theorem \ref{Thm: Unbiased Estimator, Just-identified Case} and
\begin{align*}
w_{i,s}=%
\frac{(\tilde\xi_2-\zeta_{s,2})'{M^{-1}}'(Z'Z)M^{-1}e_ie_i'(\tilde\xi_2-\zeta_{s,2})}{(\tilde\xi_2-\zeta_{s,2})'{M^{-1}}'(Z'Z)M^{-1}(\tilde\xi_2-\zeta_{s,2})}.
\end{align*}
The estimator is given by the average over $S$ simulation draws:
\begin{align*}
\hat\beta=\frac{1}{S}\sum_{i=1}^n \hat\beta_{s}.
\end{align*}
In our application, we use $S=100,000$ simulation draws.

\end{itemize}

\subsection{\citet{Hornung2014}}

\citet{Hornung2014} studies the long term impact of the flight of skilled
Huguenot refugees from France to Prussia in the seventeenth century.
He finds that regions of Prussia which received more Huguenot refugees
during the late seventeenth century had a higher level of productivity in
textile manufacturing at the start of the nineteenth century. To address
concerns over endogeneity in Huguenot settlement patterns and obtain
an estimate for the causal effect of skilled immigration on productivity,
\citet{Hornung2014} considers specifications which instrument Huguenot
immigration to a given region using population losses due to plague
at the end of the Thirty Years' War. For more information on the data
and motivation of the instrument,
see \citet{Hornung2014}.

Hornung's argument for the validity of his instrument clearly implies
that the first-stage effect should be positive, but the relationship
between the instrument and the endogenous regressors appears to be
fairly weak. In particular, the four IV specifications reported in
Tables 4 and 5 of \citet{Hornung2014} have first-stage F-statistics of 3.67,
4.79, 5.74, and 15.35, respectively. Thus, it seems that the conventional
normal approximation to the distribution of IV estimates may be unreliable
in this context. In each of the four main IV specifications considered
by Hornung, we compare 2SLS and Fuller (again with constant equal to one) to our estimator. Since
there is only a single instrument in this context, the model is just-identified
and the unbiased estimator is unique. In each specification we also
compute and report an identification-robust Anderson-Rubin confidence
set for the coefficient on the endogenous regressor. The results are
reported in Table \ref{tab:Hornung Results}.

\begin{sidewaystable}

\begin{tabular}{|c|c|c|c|c|c|c|c|c|}
\hline 
{\small{}Specification} & {\small{}Estimator} & {\small{}I} &  & {\small{}II} &  & {\small{}III} &  & {\small{}IV}\tabularnewline
\hline 
{\small{}$X:$ Percent Huguenots in 1700} & {\small{}2SLS} & {\small{}3.48} &  & {\small{}3.38} &  & {\small{}1.67} &  & \tabularnewline
\hline 
 & {\small{}Fuller} & {\small{}3.17} &  & {\small{}3.08} &  & {\small{}1.59} &  & \tabularnewline
\hline 
 & {\small{}Unbiased} & {\small{}3.24} &  & {\small{}3.14} &  & {\small{}1.61} &  & \tabularnewline
\hline 
{\small{}$X:$ log Huguenots in 1700} & {\small{}2SLS} &  &  &  &  &  &  & {\small{}0.07}\tabularnewline
\hline 
 & {\small{}Fuller} &  &  &  &  &  &  & {\small{}0.07}\tabularnewline
\hline 
 & {\small{}Unbiased} &  &  &  &  &  &  & {\small{}0.07}\tabularnewline
\hline 
{\small{}95\% AR Confidence Set} &  & {\small{}(-$\infty$,59.23{]}$\cup${[}1.55,$\infty$)} &  & {\small{}{[}1.64,19.12{]}} &  & {[}-0.45,5.93{]} &  & {\small{}{[}-0.01,0.16{]}}\tabularnewline
\hline 
{\small{}Other controls} &  & {\small{}Yes} &  & {\small{}Yes} &  & {\small{}Yes} &  & {\small{}Yes}\tabularnewline
\hline 
{\small{}Observations} &  & {\small{}150} &  & {\small{}150} &  & {\small{}186} &  & {\small{}186}\tabularnewline
\hline 
{\small{}Number of Towns} &  & {\small{}57} &  & {\small{}57} &  & {\small{}71} &  & {\small{}71}\tabularnewline
\hline 
{\small{}First Stage F-Statistic} &  & {\small{}3.67} &  & {\small{}4.79} &  & {\small{}5.74} &  & {\small{}15.35}\tabularnewline
\hline 
\end{tabular}\protect\caption{Results in \citet{Hornung2014} data. Specifications in columns I and II correspond
to Table 4 columns (3) and (5) in \citet{Hornung2014}, respectively, while
columns III and IV correspond to Table 5 columns (3) and (6) in \citet{Hornung2014}.
$Y=$log output, $X$ as indicated, and $Z=$unadjusted
population losses in I, interpolated population losses in II, and
population losses averaged over several data sources in III
and IV.  See \citet{Hornung2014}. The 2SLS  and Fuller rows report two stage least squares and Fuller estimates, respectively,
while Unbiased reports $\hat\beta_U$. Other controls
include a constant, a dummy for whether a town had relevant textile
production in 1685, measurable inputs to the production process, and
others as in \citet{Hornung2014}. As in \citet{Hornung2014}, all covariance estimates
are clustered at the town level.  Note that the unbiased and Fuller estimates, as well as the AR confidence sets, have been updated to correct an error in the March 22, 2015 version of the present paper. \label{tab:Hornung Results} }
\end{sidewaystable}

As we can see from Table \ref{tab:Hornung Results}, our unbiased
estimates in specifications I-III are smaller than the 2SLS estimates computed
in \citet{Hornung2014} (the unbiased estimate is smaller in specification IV as well, though the difference only appears in the fourth decimal place).
  Fuller estimates are, in turn, smaller than our unbiased estimates.
Nonetheless, difference between the 2SLS and unbiased estimates is less than half of the 2SLS standard error
in every specification.  In specifications I-III, where the instruments are relatively weak, the 95\%
AR confidence sets are substantially wider than 95\% confidence sets calculated using 2SLS standard errors, while in
specification IV the AR confidence set is fairly similar to the conventional 2SLS confidence set.

\subsection{\citet{angrist_does_1991}}

\citet{angrist_does_1991} are interested in the relationship between education and
labor market earnings. They argue that students born later in the
calendar year face a longer period of compulsory schooling than
those born earlier in the calendar year, and that quarter of
birth is a valid instrument for years of schooling. As we
note above their argument implies that the sign of the first-stage
effect is known. A substantial literature, beginning  with \citet{BoundJaegerBaker1995},
notes that the relationship between the instruments and the endogenous
regressor appears to be quite weak in some specifications considered
in \citet{angrist_does_1991}. Here we consider four specifications
from \citet{staiger_instrumental_1997}, based on the 1930-1939 cohort.
See \citet{angrist_does_1991}
and \citet{staiger_instrumental_1997} for more on the data
and specification.

We calculate unbiased estimators $\hat\beta_{RB}^*$, $\hat\beta_{RB,0.1}^*$, $\hat\beta_{RB,0.5}^*$, and $\hat\beta_{RB,0.9}^*$.
In all cases we take $W=Z'Z$.
To calculate confidence sets we use the 
quasi-CLR (or GMM-M) test of \citet{Kleibergen2005}, which simplifies to the CLR test of \citet{Moreira2003} under homoskedasticity 
and so delivers nearly-optimal confidence sets in that case
\citep[see][]{Mikusheva2010}.  Thus, since as discussed above the data in this application appears reasonably close to homoskedasticity, we may 
reasonably expect the quasi-CLR confidence set to perform well.
All results are reported in Table \ref{tab:Angrist Krueger Results}.
\begin{table}
\begin{tabular}{|c|c|c|c|c|c|c|c|}
\hline 
Specification & \multicolumn{1}{c|}{I} &  & \multicolumn{1}{c|}{II} &  & \multicolumn{1}{c|}{III} &  & \multicolumn{1}{c|}{IV}\tabularnewline
\hline 
 & $\hat{\beta}$ &  & $\hat{\beta}$ &  & $\hat{\beta}$ &  & $\hat{\beta}$\tabularnewline
\hline 
\hline 
2SLS & 0.099 &  & 0.081 &  & 0.060 &  & 0.081\tabularnewline
\hline 
Fuller & 0.100 &  & 0.084 &  & 0.058 &  & 0.098\tabularnewline
\hline 
LIML & 0.100 &  & 0.084 &  & 0.057 &  & 0.098\tabularnewline
\hline 
$\hat{\beta}_{RB}^*,$& 0.097 &  & 0.085 &  & -0.041 &  & 0.056\tabularnewline
\hline 
$\hat{\beta}_{RB},$ $c=0.1$ & 0.098 &  & 0.083 &  & 0.135 &  & 0.066\tabularnewline
\hline 
$\hat{\beta}_{RB},$ $c=0.5$ & 0.098 &  & 0.083 &  & 0.135 &  & 0.066\tabularnewline
\hline 
$\hat{\beta}_{RB},$ $c=0.9$ & 0.098 &  & 0.083 &  & 0.135 &  & 0.066\tabularnewline
\hline 
First Stage F & \multicolumn{1}{c|}{30.582} &  & \multicolumn{1}{c|}{4.625} &  & \multicolumn{1}{c|}{1.579} &  & \multicolumn{1}{c|}{1.823}\tabularnewline
\hline 
QCLR CS & \multicolumn{1}{c|}{{[}0.059,0.144{]}} &  & \multicolumn{1}{c|}{{[}0.046,0.127{]}} &  & \multicolumn{1}{c|}{{[}-0.588,0.668{]}} &  & \multicolumn{1}{c|}{{[}0.056,0.150{]}}\tabularnewline
\hline 
\emph{Controls} & \multicolumn{1}{c|}{} &  & \multicolumn{1}{c|}{} &  & \multicolumn{1}{c|}{} &  & \multicolumn{1}{c|}{}\tabularnewline
\hline 
Base Controls & \multicolumn{1}{c|}{Yes} &  & \multicolumn{1}{c|}{Yes} &  & \multicolumn{1}{c|}{Yes} &  & \multicolumn{1}{c|}{Yes}\tabularnewline
\hline 
Age, $\mbox{Age}^{2}$ & \multicolumn{1}{c|}{No} &  & \multicolumn{1}{c|}{No} &  & \multicolumn{1}{c|}{Yes} &  & \multicolumn{1}{c|}{Yes}\tabularnewline
\hline 
SOB & \multicolumn{1}{c|}{No} &  & \multicolumn{1}{c|}{No} &  & \multicolumn{1}{c|}{No} &  & \multicolumn{1}{c|}{Yes}\tabularnewline
\hline 
\emph{Instruments} & \multicolumn{1}{c|}{} &  & \multicolumn{1}{c|}{} &  & \multicolumn{1}{c|}{} &  & \multicolumn{1}{c|}{}\tabularnewline
\hline 
QOB & \multicolumn{1}{c|}{Yes} &  & \multicolumn{1}{c|}{Yes} &  & \multicolumn{1}{c|}{Yes} &  & \multicolumn{1}{c|}{Yes}\tabularnewline
\hline 
QOB{*}YOB & \multicolumn{1}{c|}{No} &  & \multicolumn{1}{c|}{Yes} &  & \multicolumn{1}{c|}{Yes} &  & \multicolumn{1}{c|}{Yes}\tabularnewline
\hline 
QOB{*}SOB & \multicolumn{1}{c|}{No} &  & \multicolumn{1}{c|}{No} &  & \multicolumn{1}{c|}{No} &  & \multicolumn{1}{c|}{Yes}\tabularnewline
\hline 
\# instruments & \multicolumn{1}{c|}{3} &  & \multicolumn{1}{c|}{30} &  & \multicolumn{1}{c|}{28} &  & \multicolumn{1}{c|}{178}\tabularnewline
\hline 
Observations & \multicolumn{1}{c|}{329,509} &  & \multicolumn{1}{c|}{329,509} &  & \multicolumn{1}{c|}{329,509} &  & \multicolumn{1}{c|}{329,509}\tabularnewline
\hline 
\end{tabular}.\protect\caption{\label{tab:Angrist Krueger Results}Results for \citet{angrist_does_1991}
data. Specifications as in \citet{staiger_instrumental_1997}: $Y$ =log weekly
wages, $X$=years of schooling, instruments $Z$ and exogenous controls
as indicated. QCLR is the is the quasi-CLR (or GMM-M) confidence set of \citet{Kleibergen2005}.  Unbiased estimators calculated by averaging over 100,000 draws of $\zeta$.}
\end{table}

A few points are notable from these results.  First, we see that in specifications I and II, which have the largest first stage F-statistics, the unbiased estimates are quite close to the other point estimates.  Moreover, in these specifications the choice of $c$ makes little difference.  By contrast, in specification III, where the instruments appear to be quite weak, the unbiased estimates differ substantially, with $\hat\beta_{RB}^*$ yielding a negative point estimate and $\hat\beta_{RB,c}^*$ for $c\in\{0.1,0.5,0.9\}$ yielding positive estimates substantially larger than the other estimators considered.\footnote{All unbiased estimates are calculated by averaging over 100,000 draws of $\zeta$.  For all estimates except $\hat\beta_{RB}^*$ in specification III, the residual randomness is small.  For $\hat\beta_{RB}^*$  in specification III, however, redrawing $\zeta$ yields substantially different point estimates.  This issue persists even if we increase the number of $\zeta$ draws to 1,000,000.}  A similar, though less pronounced, version of this phenomenon arises in specification IV, where unbiased estimates are smaller than those based on conventional methods and $\hat\beta_{RB}^*$ is almost 20\% smaller than estimates based on other choices of $c$.

As in the simulations there is very little difference between the estimates for $c\in\{0.1,0.5,0.9\}$.  In particular, while not exactly the same, the estimates coincide once rounded to three decimal places in all specifications.  Given that these estimators are more robust to violations of the sign restriction than that with $c=0$, we think it makes more sense to focus on these estimates.

\section{Conclusion}

In this paper, we show that a sign restriction on the first stage suffices to allow finite-sample unbiased estimation in linear IV models with normal errors and known reduced-form error covariance.  Our results suggest several avenues for further research.  First, while the focus of this paper is on estimation, recent work by 
\citet{MoreiraMillsVilela2014} finds good power for particular identification-robust conditional t-tests, suggesting that it may be interesting to consider tests based on our unbiased estimators, particularly in over-identifed contexts where the Anderson-Rubin test is no longer uniformly most powerful unbiased.  More broadly, it may be interesting to study other ways to use the knowledge of the first stage sign, both for testing and estimation purposes.

\bibliographystyle{agsm}
\bibliography{References}

\newpage{}

\noindent \begin{center}
{\large{}Appendix to the paper}
\par\end{center}{\large \par}

\noindent \begin{center}
{\LARGE{}Unbiased Instrumental Variables Estimation Under Known First-Stage Sign}
\par\end{center}{\LARGE \par}

\medskip{}

\begin{center}
{\large{} ~~~~~~~Isaiah Andrews  ~~~~~~~~~~~~~~~  Timothy B. Armstrong}
\par\end{center}{\large \par}

\noindent \begin{center}
\today
\par\end{center}

\appendix

This appendix contains proofs and additional results for the paper ``Unbiased Instrumental Variables Estimation Under Known First-Stage Sign.''
Appendix \ref{proof_append} gives proofs for results stated in the main text.  Appendix \ref{estimated_variance_sec} derives asymptotic results for models with non-normal errors and an unknown reduced-form error variance.  Appendix \ref{hp_append} relates our results to those of \citet{HiranoPoter2015}.  Appendix \ref{risk_lower_bound_sec} derives a lower bound on the risk of unbiased estimators in over-identified models, discusses cases in which the bound in attained, and proves that there is no uniformly minimum risk unbiased estimator in such models.  Appendix \ref{just-id sim append} gives additional simulation results for the just-identified case, while Appendix \ref{multi-inst sim design} details our simulation design for the over-identified case.

\section{Proofs}\label{proof_append}

This appendix contains proofs of the results in the main text.
The notation is the same as in the main text.

\subsection{Single Instrument Case}

This section proves the results from Section \ref{single_instrument_sec}, which treats the single instrument case ($k=1$).  We prove Lemma \ref{Lemma:Voinov Nikulin} and Theorems \ref{Thm: Unbiased Estimator, Just-identified Case}, \ref{betau_second_moment_thm} and \ref{Thm: Limiting Distribution, Just-identified Case}.

We first prove Lemma \ref{Lemma:Voinov Nikulin}, which shows unbiasedness of $\hat\tau$ for $1/\pi$.
As discussed in the main text, this result is known in the literature \citep[see, e.g., pp. 181-182 of ][]{VoinovNikulin1993}.  We give a constructive proof based on elementary calculus (\citeauthor{VoinovNikulin1993} provide a derivation based on the bilateral Laplace transform).

\begin{proof}[Proof of Lemma \ref{Lemma:Voinov Nikulin}]
Since $\xi_2/\sigma_2\sim N(\pi/\sigma_2,1)$, we have
\begin{align*}
&E_{\pi,\beta}\hat\tau(\xi_2,\sigma_2^2)
=\frac{1}{\sigma_2}\int \frac{1-\Phi(x)}{\phi(x)}\phi(x-\pi/\sigma_2)\, dx
=\frac{1}{\sigma_2}\int (1-\Phi(x))\exp\left((\pi/\sigma_2)x-(\pi/\sigma_2)^2/2\right)\, dx  \\
&=\frac{1}{\sigma_2}\exp(-(\pi/\sigma_2)^2/2)
\left\{\left[(1-\Phi(x))(\sigma_2/\pi)\exp((\pi/\sigma_2)x)\right]_{x=-\infty}^\infty
+\int (\sigma_2/\pi)\exp((\pi/\sigma_2)x)\phi(x)\, dx
\right\},
\end{align*}
using integration by parts to obtain the last equality.  Since the first term in brackets in the last line is zero, this is equal to
\begin{align*}
&\frac{1}{\sigma_2}\int (\sigma_2/\pi)\exp((\pi/\sigma_2)x-(\pi/\sigma_2)^2/2)\phi(x)\, dx
=\frac{1}{\pi}\int \phi(x-\pi/\sigma_2)\, dx=\frac{1}{\pi}.
\end{align*}

\end{proof}

We note that $\hat\tau$ has an infinite $1+\varepsilon$ moment for $\varepsilon>0$.

\begin{lemma}\label{tau_second_moment_lemma}
The expectation of $\hat\tau(\xi_2,\sigma_2^2)^{1+\varepsilon}$ is infinite for all $\pi$ and $\varepsilon>0$.
\end{lemma}
\begin{proof}
By similar calculations to those in the proof of Lemma \ref{Lemma:Voinov Nikulin},
\begin{align*}
E_{\pi,\beta}\hat\tau(\xi_2,\sigma_2^2)^{1+\varepsilon}
=\frac{1}{\sigma_2^{1+\varepsilon}}\int \frac{(1-\Phi(x))^{1+\varepsilon}}{\phi(x)^\varepsilon}\exp\left((\pi/\sigma_2)x-(\pi/\sigma_2)^2/2\right)\, dx.
\end{align*}
For $x<0$, $1-\Phi(x)\ge 1/2$, so the integrand is bounded from below by a constant times
$\exp(\varepsilon x^2/2+(\pi/\sigma_2)x)$, which is bounded away from zero as $x\to-\infty$.

\end{proof}

\begin{proof}[Proof of Theorem \ref{Thm: Unbiased Estimator, Just-identified Case}]
To establish unbiasedness, note that since $\xi_2$ and $\xi_1-\frac{\sigma_{12}}{\sigma_2^2}\xi_2$ are jointly normal with zero covariance, they are independent.  Thus,
\begin{align*}
E_{\pi,\beta}\hat\beta_U(\xi,\Sigma)
=\left(E_{\pi,\beta}\hat\tau\right)\left[E_{\pi,\beta}\left(\xi_1-\frac{\sigma_{12}}{\sigma_2^2}\xi_2\right)\right]
+\frac{\sigma_{12}}{\sigma_2^2}
=\frac{1}{\pi}\left(\pi\beta-\frac{\sigma_{12}}{\sigma_{22}}\pi\right)+\frac{\sigma_{12}}{\sigma_{22}}
=\beta
\end{align*}
since $E_{\pi,\beta}\hat\tau=1/\pi$ by Lemma \ref{Lemma:Voinov Nikulin}.

To establish uniqueness,  consider any unbiased estimator $\hat{\beta}\left(\xi,\Sigma\right)$.
By unbiasedness
\begin{align*}
E_{\pi,\beta}\left[\hat{\beta}\left(\xi,\Sigma\right)-\hat\beta_U(\xi,\Sigma)\right]=0~\forall\beta\in B,\pi\in\Pi.
\end{align*}
The parameter space contains an open set by assumption, so by Theorem 4.3.1 of \citet{lehmann_testing_2005}
the family of distributions of $\xi$ under $(\pi,\beta)\in\Theta$ is complete.  Thus $\hat{\beta}\left(\xi,\Sigma\right)-\hat\beta_U(\xi,\Sigma)=0$
almost surely for all $(\pi,\beta)\in\Theta$ by the definition of completeness.

\end{proof}

\begin{proof}[Proof of Theorem \ref{betau_second_moment_thm}]
If $E_{\pi,\beta}\left|\hat\beta_U(\xi,\Sigma)\right|^{1+\varepsilon}$ were finite, then $E_{\pi,\beta}\left|\hat\beta_U(\xi,\Sigma)-\sigma_{12}/\sigma_2^2\right|^{1+\varepsilon}$ would be finite as well by Minkowski's inequality.  But
\begin{align*}
E_{\pi,\beta}\left|\hat\beta_U(\xi,\Sigma)-\sigma_{12}/\sigma_2^2\right|^{1+\varepsilon}
=E_{\pi,\beta}\left|\hat{\tau}\left(\xi_{2},\sigma_{2}^{2}\right)\right|^{1+\varepsilon}
E_{\pi,\beta}\left|\xi_1-\frac{\sigma_{12}}{\sigma_{22}^2}\xi_2\right|^{1+\varepsilon},
\end{align*}
and the second term is nonzero since $\Sigma$ is positive definite.  Thus, the $1+\varepsilon$ absolute moment is infinite by Lemma \ref{tau_second_moment_lemma}.  The claim that any unbiased estimator has infinite $1+\varepsilon$ moment follows from Rao-Blackwell: since $\hat\beta_U(\xi,\Sigma)=E\left[\tilde\beta(\xi,\Sigma)|\xi\right]$ for any unbiased estimator $\tilde\beta$ by the uniqueness of the non-randomized unbiased estimator based on $\xi$, Jensen's inequality implies that the $1+\varepsilon$ moment of $|\tilde\beta|$ is bounded from below by the (infinite) $1+\varepsilon$ moment of $|\hat\beta_U|$.
\end{proof}

We now consider the behavior of $\hat\beta_U$ relative to the usual 2SLS estimator (which, in the single instrument case considered here, is given by $\hat\beta_{2SLS}=\xi_1/\xi_2$) as $\pi\to\infty$.

\begin{proof}[Proof of Theorem \ref{Thm: Limiting Distribution, Just-identified Case}]
Note that
\begin{align*}
&\hat\beta_U-\hat\beta_{2SLS}
=\left(\hat\tau(\xi_2,\sigma_2^2)-\frac{1}{\xi_2}\right)\left(\xi_1-\frac{\sigma_{12}}{\sigma_2^2}\xi_2\right)
=\left(\xi_2\hat\tau(\xi_2,\sigma_2^2)-1\right)\left(\frac{\xi_1}{\xi_2}-\frac{\sigma_{12}}{\sigma_2^2}\right).
\end{align*}
As $\pi\to\infty$, $\xi_1/\xi_2=\hat\beta_{2SLS}=\mathcal{O}_P(1)$, so it suffices to show that $\pi\left(\xi_2\hat\tau(\xi_2,\sigma_2^2)-1\right)=o_P(1)$ as $\pi\to\infty$.  Note that, by Section 2.3.4 of \citet{Small2010},
\begin{align*}
\pi\left|\xi_2\hat\tau(\xi_2,\sigma_2^2)-1\right|
=\pi\left|\frac{\xi_2}{\sigma_2}\frac{1-\Phi(\xi_2/\sigma_2)}{\phi(\xi_2/\sigma_2)}-1\right|
\le \pi\frac{\sigma_2^2}{\xi_2^2}
= \frac{\pi}{\xi_2}\frac{\sigma_2^2}{\xi_2}.
\end{align*}
This converges in probability to zero since $\pi/\xi_2\stackrel{p}{\to} 1$ and $\frac{\sigma_2^2}{\xi_2}\stackrel{p}{\to} 0$ as $\pi\to\infty$.
\end{proof}

The following lemma regarding the mean absolute deviation of $\hat\beta_U$ will be useful in the next section treating the case with multiple instruments.

\begin{lemma}\label{mad_bound_lemma}
For a constant $K(\beta,\Sigma)$ depending only on $\Sigma$ and $\beta$ (but not on $\pi$),
\begin{align*}
\pi E_{\pi,\beta}\left|\hat\beta_U(\xi,\Sigma)-\beta\right|\le K(\beta,\Sigma).
\end{align*}
\end{lemma}
\begin{proof}
We have
\begin{align*}
&\pi \left(\hat\beta_U-\beta\right)
=\pi\left[\hat\tau\cdot\left(\xi_1-\frac{\sigma_{12}}{\sigma_2^2}\xi_2\right)+\frac{\sigma_{12}}{\sigma_2^2}-\beta\right]
=\pi \hat\tau\cdot\left(\xi_1-\frac{\sigma_{12}}{\sigma_2^2}\xi_2\right)
+\pi\frac{\sigma_{12}}{\sigma_2^2}-\pi\beta  \\
&=\pi \hat\tau\cdot\left(\xi_1-\beta\pi-\frac{\sigma_{12}}{\sigma_2^2}(\xi_2-\pi)\right)
+\pi \hat\tau \beta\pi-\pi\hat\tau \frac{\sigma_{12}}{\sigma_2^2}\pi
+\pi\frac{\sigma_{12}}{\sigma_2^2}-\pi\beta  \\
&=\pi \hat\tau\cdot\left(\xi_1-\beta\pi-\frac{\sigma_{12}}{\sigma_2^2}(\xi_2-\pi)\right)
+\pi(\pi \hat\tau-1)\left(\beta-\frac{\sigma_{12}}{\sigma_2^2}\right).
\end{align*}
Using this and the fact that $\xi_2$ and $\xi_1-\frac{\sigma_{12}}{\sigma_{2}^2}\xi_2$ are independent, it follows that
\begin{align*}
\pi E_{\pi,\beta}\left|\hat\beta_U-\beta\right|
\le
E_{\pi,\beta}\left|\xi_1-\beta\pi-\frac{\sigma_{12}}{\sigma_2^2}(\xi_2-\pi)\right|
+\pi E_{\pi,\beta}|\pi \hat\tau-1|\left|\beta-\frac{\sigma_{12}}{\sigma_2^2}\right|,
\end{align*}
where we have used the fact that $E_{\pi,\beta} \pi \hat\tau =1$.
The only term in the above expression that depends on $\pi$ is
$\pi E_{\pi,\beta}|\pi \hat\tau-1|$.  Note that this is bounded above by
$\pi E_{\pi,\beta}\pi \hat\tau+\pi=2\pi$ (using unbiasedness and positivity of $\hat\tau$), so we can assume an arbitrary lower bound on $\pi$ when bounding this term.

Letting $\tilde\pi=\pi/\sigma_2$, we have $\xi_2/\sigma_2\sim N(\tilde\pi,1)$, so that
\begin{align*}
\frac{\pi}{\sigma_2}E_{\pi,\beta}|\pi \hat\tau-1|
=\frac{\pi}{\sigma_2}E_{\pi,\beta}\left|\frac{\pi}{\sigma_2}\frac{1-\Phi(\xi_2/\sigma_2)}{\phi(\xi_2/\sigma_2)}-1\right|
=\tilde\pi \int \left|\tilde\pi\frac{1-\Phi(z)}{\phi(z)}-1\right|\phi(z-\tilde\pi)\, dz.
\end{align*}
Let $\varepsilon>0$ be a constant to be determined later in the proof.
By (1.1) in \citet{baricz_mills_2008}
\begin{align*}
&\tilde \pi^2 \int_{z\ge \tilde \pi\varepsilon} \left|\frac{1-\Phi(z)}{\phi(z)}-\frac{1}{\tilde \pi}\right|
\phi(z-\tilde \pi)\, dz  \\
&\le \tilde \pi^2 \int_{z\ge \tilde \pi\varepsilon} \left|\frac{1}{z}-\frac{1}{\tilde \pi}\right|
\phi(z-\tilde \pi)\, dz
+ \tilde \pi^2 \int_{z\ge \tilde \pi\varepsilon} \left|\frac{z}{z^2+1}-\frac{1}{\tilde \pi}\right|
\phi(z-\tilde \pi)\, dz.
\end{align*}
The first term is
\begin{align*}
\tilde \pi^2 \int_{z\ge \tilde \pi\varepsilon} \left|\frac{\tilde \pi-z}{\tilde\pi z}\right|
\phi(z-\tilde \pi)\, dz
\le \tilde \pi^2 \int_{z\ge \tilde \pi\varepsilon} \left|\frac{\tilde \pi-z}{\tilde\pi^2\varepsilon}\right|
\phi(z-\tilde \pi)\, dz
\le\frac{1}{\varepsilon}\int |u|\phi(u)\, du.
\end{align*}
The second term is
\begin{align*}
&\tilde \pi^2 \int_{z\ge \tilde \pi\varepsilon} \left|\frac{1}{z+1/z}-\frac{1}{\tilde \pi}\right|
\phi(z-\tilde \pi)\, dz
=\tilde \pi^2 \int_{z\ge \tilde \pi\varepsilon} \left|\frac{\tilde \pi-(z+1/z)}{\tilde \pi (z+1/z)}\right|
\phi(z-\tilde \pi)\, dz  \\
&\le \tilde \pi^2 \int_{z\ge \tilde \pi\varepsilon} \frac{\left|\tilde \pi-z\right|+\frac{1}{\varepsilon\tilde\pi}}{\tilde \pi^2\varepsilon}
\phi(z-\tilde \pi)\, dz
\le\frac{1}{\varepsilon}\int
\left(\left|u\right|+\frac{1}{\varepsilon\tilde\pi}\right)
\phi(u)\, dz.
\end{align*}
We also have
\begin{align*}
&\tilde \pi^2 \int_{z< \tilde \pi\varepsilon} \left|\frac{1-\Phi(z)}{\phi(z)}-\frac{1}{\tilde \pi}\right|
\phi(z-\tilde \pi)\, dz
\le \tilde \pi^2 \int_{z< \tilde \pi\varepsilon} \frac{1-\Phi(z)}{\phi(z)}
\phi(z-\tilde \pi)\, dz
+\tilde \pi \int_{z< \tilde \pi\varepsilon} 
\phi(z-\tilde \pi)\, dz.
\end{align*}
The second term is equal to $\tilde\pi \Phi(\tilde \pi\varepsilon-\tilde\pi)$, which is bounded uniformly over $\tilde\pi$ for $\varepsilon<1$.
The first term is
\begin{align*}
&\tilde \pi^2 \int_{z< \tilde \pi\varepsilon} (1-\Phi(z))
  \exp\left(\tilde\pi z-\frac{1}{2}\tilde\pi^2\right)\, dz  \\
&=\tilde \pi^2 \int_{z< \tilde \pi\varepsilon}
\int_{t\ge z}\phi(t)
  \exp\left(\tilde\pi z-\frac{1}{2}\tilde\pi^2\right)\,dt dz  \\
&=\tilde \pi^2 \int_{t\in\mathbb{R}}
\int_{z\le \min\{t,\tilde\pi\varepsilon\}}\phi(t)
  \exp\left(\tilde\pi z-\frac{1}{2}\tilde\pi^2\right)\,dz dt  \\
&=\tilde \pi^2\exp\left(-\frac{1}{2}\tilde\pi^2\right)
\int_{t\in\mathbb{R}}
\phi(t)\left[\frac{1}{\tilde\pi}\exp\left(\tilde\pi z\right)\right]_{z=-\infty}^{\min\{t,\tilde\pi\varepsilon\}}\,dt \\
&=\tilde \pi\exp\left(-\frac{1}{2}\tilde\pi^2\right)
\int_{t\in\mathbb{R}}
\phi(t)\exp\left(\tilde\pi\min\{t,\tilde\pi\varepsilon\}\right)\,dt  \\
&\le \tilde \pi\exp\left(-\frac{1}{2}\tilde\pi^2+\varepsilon \tilde\pi^2\right).
\end{align*}
For $\varepsilon<1/2$, this is uniformly bounded over all $\tilde\pi>0$.

\end{proof}

\subsection{Multiple Instrument Case}\label{multi_instrument_append}

This section proves Theorem \ref{rb_asym_efficiency_thm} and extends this theorem to cover unbiased estimators that are efficient under strong instrument asymptotics in the heteroskedastic case.  In particular, we prove an extension of this theorem allowing for unbiased estimators that are asymptotically equivalent to a GMM estimator of the form
$\hat\beta_{GMM,W}=\frac{\xi_2'\hat W\xi_1}{\xi_2'\hat W\xi_2}$,
where $\hat W=\hat W(\xi)$ is a data dependent weighting matrix.
For Theorem \ref{rb_asym_efficiency_thm}, $\hat W$ is the deterministic matrix $Z'Z$.
In models with non-homoskedastic errors the two step GMM estimator with weighting matrix
\begin{align}\label{2step_gmm_eq}
\hat W=\left(\Sigma_{11}-\hat\beta_{2SLS} (\Sigma_{12}+\Sigma_{21})+\hat\beta_{2SLS}^2\Sigma_{22}\right)^{-1}
\end{align}
is asymptotically efficient under strong instruments.
Here, $\hat W$ is an estimate of the inverse of the variance matrix of the moments $\xi_1-\beta \xi_2$, which the GMM estimator sets close to zero.  Let
\begin{align}\label{2step_gmm_est_weight_eq}
\hat w_{GMM,i}^*(\xi^{(b)})
=\frac{{\xi_2^{(b)}}'\hat W(\xi^{(b)}) e_ie_i' {\xi_2^{(b)}}}{{\xi_2^{(b)}}'\hat W(\xi^{(b)}) {\xi_2^{(b)}}}
\end{align}
where
\begin{align*}
\hat W(\xi^{(b)})=\left(\Sigma_{11}-\hat\beta(\xi^{(b)}) (\Sigma_{12}+\Sigma_{21})+\hat\beta(\xi^{(b)})^2\Sigma_{22}\right)^{-1}
\end{align*}
for a preliminary estimator $\hat\beta(\xi^{(b)})$ of $\beta$ based on $\xi^{(b)}$.  The Rao-Blackwellized estimator formed by replacing $\hat w^*$ with $\hat w^*_{GMM}$ in the definition of $\hat\beta_{RB}^*$ gives an unbiased estimator that is asymptotically efficient under strong instrument asymptotics with non-homoskedastic errors,
as we now show by proving an extension of Theorem \ref{rb_asym_efficiency_thm} that covers the weight matrix in (\ref{2step_gmm_eq}) in addition to the matrix $Z'Z$ used in Theorem \ref{rb_asym_efficiency_thm}.

Consider the GMM estimator 
$\hat\beta_{GMM,W}=\frac{\xi_2'\hat W\xi_1}{\xi_2'\hat W\xi_2}$,
where $\hat W=\hat W(\xi)$ is a data dependent weighting matrix.
For Theorem \ref{rb_asym_efficiency_thm}, $\hat W$ is the deterministic matrix $Z'Z$ while, in the extension discussed above,
$\hat W$ is defined in (\ref{2step_gmm_eq}).
In both cases, $\hat W\stackrel{p}{\to} W^*$ for some positive definite matrix $W^*$ under the strong instrument asymptotics in the theorem.  For this $W^*$, define the oracle weights
\begin{align*}
w_{i}^*=\pi_i\frac{\pi'W^* e_i}{\pi'W^* \pi}
=\frac{\pi'W^* e_ie_i'\pi}{\pi'W^* \pi}
\end{align*}
and the oracle estimator
\begin{align*}
\hat\beta_{RB}^o=\hat\beta_{RB}(\xi,\Sigma;w^*)
=\hat\beta_{w}(\xi,\Sigma;w^*)
=\sum_{i=1}^kw_{i}^*\hat\beta_{U}(\xi(i),\Sigma(i)).
\end{align*}
Define the estimated weights as in (\ref{2step_gmm_est_weight_eq}):
\begin{align*}
\hat w_{i}^*=\hat w_{i}^*(\xi^{(b)})
=\frac{{\xi_2^{(b)}}'\hat W(\xi^{(b)}) e_ie_i'\xi_2^{(b)}}{{\xi_2^{(b)}}'\hat W(\xi^{(b)}) \xi_2^{(b)}}
\end{align*}
and the Rao-Blackwellized estimator based on the estimated weights
\begin{align*}
\hat\beta_{RB}^*=\hat\beta_{RB}(\xi,\Sigma;\hat w^*)
=E\left[\hat\beta_{w}(\xi^{(a)},2\Sigma;\hat w^*)\Big|\xi\right]
=\sum_{i=1}^k E\left[\hat w_{i}^*(\xi_2^{(b)})\hat\beta_{U}(\xi^{(a)}(i),2\Sigma(i))\Big|\xi\right].
\end{align*}
In the general case, we will assume that
$\hat w_{i}^*(\xi^{(b)})$ is uniformly bounded (this holds for equivalence with 2SLS under the conditions of Theorem \ref{rb_asym_efficiency_thm}, since $\sup_{\|u\|_{Z'Z}=1}u' Z'Z e_ie_i'u$ is bounded, and one can likewise show that it holds for two step GMM provided $\Sigma$ has full rank).
Let us also define the oracle linear combination of 2SLS estimators
\begin{align*}
\hat\beta_{2SLS}^o=\sum_{i=1}^k w_{i}^*\frac{\xi_{1,i}}{\xi_{2,i}}.
\end{align*}

\begin{lemma}\label{deterministic_w_lemma}
Suppose that $\hat w$ is deterministic: $\hat w(\xi^{(b)})=w$ for some constant vector $w$.  Then $\hat\beta_{RB}(\xi,\Sigma;w)=\hat\beta_{w}(\xi,\Sigma;w)$.
\end{lemma}
\begin{proof}
We have
\begin{align*}
\hat\beta_{RB}(\xi,\Sigma;w)
=E\left[\sum_{i=1}^kw_{i}\hat\beta_{U}(\xi^{(a)}(i),2\Sigma(i))\bigg|\xi\right]
=\sum_{i=1}^kw_{i}E\left[\hat\beta_{U}(\xi^{(a)}(i),2\Sigma(i))\bigg|\xi\right].
\end{align*}
Since $\xi^{(a)}(i)=\zeta(i)+\xi(i)$ (where $\zeta(i)=(\zeta_i,\zeta_{k+i})'$), $\xi^{(a)}(i)$ is independent of $\{\xi(j)\}_{j\ne i}$ conditional on $\xi(i)$.  Thus,
$E\left[\hat\beta_{U}(\xi^{(a)}(i),2\Sigma(i))\bigg|\xi\right]
=E\left[\hat\beta_{U}(\xi^{(a)}(i),2\Sigma(i))\bigg|\xi(i)\right]$.
Since $E\left[\hat\beta_{U}(\xi^{(a)}(i),2\Sigma(i))\bigg|\xi(i)\right]$ is an unbiased estimator for $\beta$ that is a deterministic function of $\xi(i)$, it must be equal to $\hat\beta_U(\xi(i),\Sigma(i))$, the unique nonrandom unbiased estimator based on $\xi(i)$ (where uniqueness follows by completeness since the parameter space $\{(\beta\pi_i,\pi_i)|\pi_i\in\mathbb{R}_+,\beta\in\mathbb{R}\}$ contains an open rectangle).
Plugging this in to the above display gives the result.
\end{proof}

\begin{lemma}
Let $\|\pi\|\to\infty$ with $\|\pi\|/\min_i\pi_i=\mathcal{O}(1)$.  Then
$\|\pi\|\left(\hat\beta_{GMM,W}-\hat\beta_{2SLS}^o\right)\stackrel{p}{\to} 0$.
\end{lemma}
\begin{proof}
Note that
\begin{align*}
&\hat\beta_{GMM,W}-\hat\beta_{2SLS}^o
=\frac{\xi_2'\hat W\xi_1}{\xi_2'\hat W\xi_2}-\sum_{i=1}^k w_{i}^*\frac{\xi_{1,i}}{\xi_{2,i}}
=\sum_{i=1}^k \left(\frac{\xi_2'\hat W e_ie_i'\xi_2}{\xi_2'\hat W\xi_2}-w_{i}^*\right)\frac{\xi_{1,i}}{\xi_{2,i}}  \\
&=\sum_{i=1}^k \left(\frac{\xi_2'\hat W e_ie_i'\xi_2}{\xi_2'\hat W\xi_2}-\frac{\pi'W^* e_ie_i'\pi}{\pi'W^*\pi}\right)\frac{\xi_{1,i}}{\xi_{2,i}}
=\sum_{i=1}^k \left(\frac{\xi_2'\hat W e_ie_i'\xi_2}{\xi_2'\hat W\xi_2}-\frac{\pi'W^* e_ie_i'\pi}{\pi'W^*\pi}\right)\left(\frac{\xi_{1,i}}{\xi_{2,i}}-\beta\right),
\end{align*}
where the last equality follows since 
$\sum_{i=1}^k \frac{\xi_2'\hat W e_ie_i'\xi_2}{\xi_2'\hat W\xi_2}=\sum_{i=1}^k \frac{\pi'W^* e_ie_i'\pi}{\pi'W^*\pi}=1$ with probability one.  For each $i$,
$\pi_i(\xi_{1,i}/\xi_{2,i}-\beta)=\mathcal{O}_P(1)$
and $\frac{\xi_2'\hat W e_ie_i'\xi_2}{\xi_2'\hat W\xi_2}-\frac{\pi'W^* e_ie_i'\pi}{\pi'W^*\pi}\stackrel{p}{\to} 0$ as the elements of $\pi$ approach infinity.  Combining this with the above display and the fact that $\|\pi\|/\min_i\pi_i=\mathcal{O}(1)$ gives the result.
\end{proof}

\begin{lemma}
Let $\|\pi\|\to\infty$ with $\|\pi\|/\min_i\pi_i=\mathcal{O}(1)$.  Then
$\|\pi\|\left(\hat\beta_{2SLS}^o-\hat\beta_{RB}^o\right)\stackrel{p}{\to} 0$.
\end{lemma}
\begin{proof}
By Lemma \ref{deterministic_w_lemma},
\begin{align*}
\|\pi\|\left(\hat\beta_{2SLS}^o-\hat\beta_{RB}^o\right)
=\|\pi\|\sum_{i=1}^kw_{i}^*\left(\frac{\xi_{1,i}}{\xi_{2,i}}-\hat\beta_U(\xi(i),\Sigma(i))\right).
\end{align*}
By Theorem \ref{Thm: Limiting Distribution, Just-identified Case},
$\pi_i\left(\frac{\xi_{1,i}}{\xi_{2,i}}-\hat\beta_U(\xi(i),\Sigma(i))\right)\stackrel{p}{\to} 0$.  Combining this with the boundedness of $\|\pi\|/\min_i\pi_i$ gives the result.
\end{proof}

\begin{lemma}
Let $\|\pi\|\to\infty$ with $\|\pi\|/\min_i\pi_i=\mathcal{O}(1)$.  Then
$\|\pi\|\left(\hat\beta_{RB}^o-\hat\beta_{RB}^*\right)\stackrel{p}{\to} 0$.
\end{lemma}
\begin{proof}
We have
\begin{align*}
&\hat\beta_{RB}^o-\hat\beta_{RB}^*
=\sum_{i=1}^kE\left[\left(w_{i}^*-\hat w_{i}^*(\xi^{(b)})\right)
\hat\beta_U(\xi^{(a)}(i),2\Sigma(i))\Big|\xi\right]  \\
&=\sum_{i=1}^kE\left[\left(w_{i}^*-\hat w_{i}^*(\xi^{(b)})\right)
\left(\hat\beta_U(\xi^{(a)}(i),2\Sigma(i))-\beta\right)\Big|\xi\right]
\end{align*}
using the fact that $\sum_{i=1}^kw_{i}^*=\sum_{i=1}^k\hat w_{i}^*(\xi^{(b)})=1$ with probability one.  Thus,
\begin{align*}
&E_{\beta,\pi}\left|\hat\beta_{RB}^o-\hat\beta_{RB}^*\right|
\le \sum_{i=1}^kE_{\beta,\pi}\left|\left(w_{i}^*-\hat w_{i}^*(\xi^{(b)})\right)
\left(\hat\beta_U(\xi^{(a)}(i),2\Sigma(i))-\beta\right)\right|  \\
&=\sum_{i=1}^kE_{\beta,\pi}\left|w_{i}^*-\hat w_{i}^*(\xi^{(b)})\right|
E_{\beta,\pi}\left|\hat\beta_U(\xi^{(a)}(i),2\Sigma(i))-\beta\right|.
\end{align*}
As $\|\pi\|\to\infty$, $\hat w_{i}^*(\xi^{(b)})-w_i^*\stackrel{p}{\to} 0$ so,
since $\hat w_{i}^*(\xi^{(b)})$ is bounded,
$E_{\beta,\pi}\left|w_{i}^*-\hat w_{i}^*(\xi^{(b)})\right|\to 0$.
Thus, it suffices to show that
$\pi_iE_{\beta,\pi}\left|\hat\beta_U(\xi^{(a)}(i),2\Sigma(i))-\beta\right|=\mathcal{O}(1)$
for each $i$.  But this follows by Lemma \ref{mad_bound_lemma}, which completes the proof.
\end{proof}

\section{Non-Normal Errors and Unknown Reduced Form Variance}\label{estimated_variance_sec}

This appendix derives asymptotic results for the case with non-normal errors and an estimated reduced form covariance matrix.  Section \ref{weak_iv_nonnormal_sec} shows asymptotic unbiasedness in the weak instrument case.  Section \ref{strong_iv_nonnormal_sec} shows asymptotic equivalence with 2SLS in the strong instrument case (where, in the case with multiple instruments, the weights are chosen appropriately).  %
The results are proved using some auxiliary lemmas, which are stated and proved in Section \ref{auxiliary_lemmas_sec}.

Throughout this appendix, we consider a sequence of reduced form estimators
\begin{align*}
\tilde\xi=\left(\begin{array}{c}
(Z'Z)^{-1}Z'Y  \\
(Z'Z)^{-1}Z'X  \\
\end{array}\right),
\end{align*}
which we assume satisfy a central limit theorem:
\begin{align}\label{rf_clt_eq}
\sqrt{T}\left(\tilde \xi
-\left(\begin{array}{c}
\pi_T\beta  \\
\pi_T  \\
\end{array}\right)\right)
\stackrel{d}{\to} N\left(0,\Sigma^*\right),
\end{align}
where $\pi_T$ is a sequence of parameter values and $\Sigma^*$ is a positive definite matrix.  Following \citet{staiger_instrumental_1997}, we distinguish between the case of weak instruments, in which $\pi_T$ converges to $0$ at a $\sqrt{T}$ rate, and the case of strong instruments, in which $\pi_T$ converges to a vector in the interior of the positive orthant.  Formally, the weak instrument case is given by the condition that
\begin{align}\label{weak_iv_cond}
\sqrt{T}\pi_T\to \pi^* \text{ where } \pi_i^*>0 \text{ for all } i
\end{align}
while the strong instrument case is given by the condition that
\begin{align}\label{strong_iv_cond}
\pi_T\to \pi^* \text{ where } \pi_i^*>0 \text{ for all } i.
\end{align}
In both cases, we assume the availability of a consistent estimator $\tilde \Sigma$ for the asymptotic variance of the reduced form estimators:
\begin{align}\label{Sigma_consistency_cond}
\tilde\Sigma\stackrel{p}{\to} \Sigma^*.
\end{align}
The estimator is then formed as
\begin{align*}
&\hat\beta_{RB}(\tilde\xi,\tilde\Sigma/T,\hat w)
=E_{\tilde\Sigma/T}\left[
\hat \beta_w(\tilde\xi^{(a)},2\tilde\Sigma/T,\hat w(\tilde\xi^{(b)}))\Big|\tilde\xi
\right]  \\
&=\int
\hat \beta_w(\tilde\xi+T^{-1/2}\tilde\Sigma^{1/2}\eta,2\tilde\Sigma/T,\hat w(\tilde\xi-T^{-1/2}\tilde\Sigma^{1/2}\eta))\,
dP_{N(0,I_{2k})}(\eta)
\end{align*}
where
$\tilde\xi^{(a)}=\tilde\xi+T^{-1/2}\tilde\Sigma^{1/2}\eta$
and
$\tilde\xi^{(b)}=\tilde\xi-T^{-1/2}\tilde\Sigma^{1/2}\eta$
for $\eta\sim N(0,I_{2k})$ independent of $\tilde\xi$ and $\tilde\Sigma$, and we use the subscript in the expectation to denote the dependence of the conditional distribution of $\tilde\xi^{(a)}$ and $\tilde\xi^{(b)}$ on $\tilde\Sigma/T$.  In the single instrument case, $\hat\beta_{RB}(\tilde\xi,\tilde\Sigma/T,\hat w)$ reduces to $\hat\beta_{U}(\tilde\xi,\tilde\Sigma/T)$.

For the weights $\hat w$, we assume that
$\hat w(\xi^{(b)})$ is bounded and continuous in $\xi^{(b)}$ with 
$\sum_{i=1}^k \hat w_i(\xi^{(b)})=1$
and
$\hat w_i(a\xi^{(b)})=\hat w_i(\xi^{(b)})$
 for any scalar $a$, as holds for all the weights discussed above.
Using the fact that $\hat\beta_U(\sqrt{a} x, a \Omega)=\hat\beta_U(x, \Omega)$ for any scalar $a$ and any $x$ and $\Omega$, we have,
under the above conditions on $\hat w$,
\begin{align*}
\hat\beta_{RB}(\tilde\xi,\tilde\Sigma/T,\hat w)
=\int
\hat \beta_w(\sqrt{T}\tilde\xi+\tilde\Sigma^{1/2}\eta,2\tilde\Sigma,\hat w(\sqrt{T}\tilde\xi-\tilde\Sigma^{1/2}\eta))\,
dP_{N(0,I_{2k})}(\eta)
=\hat\beta_{RB}(\sqrt{T}\tilde\xi,\tilde\Sigma,\hat w).
\end{align*}
Thus, we can focus on the behavior of $\sqrt{T}\tilde\xi$ and $\tilde\Sigma$, which are asymptotically nondegenerate in the weak instrument case.

\subsection{Weak Instrument Case}\label{weak_iv_nonnormal_sec}

The following theorem shows that the estimator $\hat\beta_{RB}$ converges in distribution to a random variable with mean $\beta$.  Note that, since convergence in distribution does not imply convergence of moments, this does not imply that the bias of $\hat\beta_{RB}$ converges to zero.  While it seems likely this stronger form of asymptotic unbiasedness could be achieved under further conditions by truncating $\hat\beta_{RB}$ at a slowly increasing sequence of points, we leave this extension for future research.

\begin{theorem}
Let 
(\ref{rf_clt_eq})
(\ref{weak_iv_cond})
and
(\ref{Sigma_consistency_cond})
hold, and suppose that
$\hat w(\xi^{(b)})$ is bounded and continuous in $\xi^{(b)}$
with
$\hat w_i(a\xi^{(b)})=\hat w_i(\xi^{(b)})$ for any scalar $a$.  Then
\begin{align*}
\hat\beta_{RB}(\tilde\xi,\tilde\Sigma/T,\hat w)
=\hat\beta_{RB}(\sqrt{T}\tilde\xi,\tilde\Sigma,\hat w)
\stackrel{d}{\to} \hat\beta_{RB}(\xi^*,\Sigma^*,\hat w)
\end{align*}
where $\xi^*\sim N(({\pi^*}'\beta,{\pi^*}')',\Sigma^*)$ and
$E\left[\hat\beta_{RB}(\xi^*,\Sigma^*,\hat w)\right]=\beta$.
\end{theorem}
\begin{proof}
Since $\sqrt{T}\tilde\xi\stackrel{d}{\to} \xi^*$ and $\tilde\Sigma\stackrel{p}{\to}\Sigma^*$, the first display follows by the continuous mapping theorem so long as $\hat\beta_{RB}(\xi^*,\Sigma^*,\hat w)$ is continuous in $\xi^*$ and $\Sigma^*$.  Since
\begin{align}\label{betarb_unif_int_eq}
\hat\beta_{RB}(\xi^*,\Sigma^*,\hat w)
=\int \hat \beta_w(\xi^*+{\Sigma^*}^{1/2}\eta,2\Sigma^*,\hat w(\xi^*-{\Sigma^*}^{1/2}\eta))\,
dP_{N(0,I_{2k})}(\eta)
\end{align}
and the integrand is continuous in $\xi^*$ and $\Sigma^*$, it suffices to show uniform integrability over $\xi^*$ and $\Sigma^*$ in an arbitrarily small neighborhood of any point.  The $p$th moment of the integrand in the above display is bounded by a constant times the sum over $i$ of
\begin{align*}
\int\int \left|\hat \beta_U(\xi^*(i)+{\Sigma^*}^{1/2}(i)z,2\Sigma^*(i))\right|^p
\phi(z_1)\phi(z_2)
\, dz_1 dz_2
=R(\xi^*(i),\Sigma^*(i),0,p),
\end{align*}
where $R$ is defined below in Section \ref{auxiliary_lemmas_sec}.  By Lemma \ref{R_tildeR_lemma} below, this is equal to
\begin{align*}
\tilde R\left(\frac{\xi_2^*(i)}{\sqrt{\Sigma_{22}^*(i)}},
\frac{\xi_1^*(i)}{\sqrt{\Sigma_{22}^*(i)}},
\left(\frac{\Sigma_{11}^*(i)}{\Sigma_{22}^*(i)}-\frac{{\Sigma_{12}^*(i)}^2}{{\Sigma_{22}^*(i)}^2}\right)^{1/2},
-\frac{\Sigma_{12}^*(i)}{\Sigma_{22}^*(i)},p\right),
\end{align*}
which is bounded uniformly over a small enough neighborhood of any $\xi^*$ and $\Sigma^*$ with $\Sigma^*$ positive definite by Lemma \ref{tildeR_bound_lemma} below so long as $p<2$.  Setting $1<p<2$, it follows that uniform integrability holds for (\ref{betarb_unif_int_eq}) so that $\hat\beta_{RB}(\xi^*,\Sigma^*,\hat w)$ is continuous, thereby giving the result.
\end{proof}

\subsection{Strong Instrument Asymptotics}\label{strong_iv_nonnormal_sec}

Let $\tilde W(\tilde\xi^{(b)},\tilde \Sigma)$
and $\hat W$ be weighting matrices that converge in probability to some positive definite symmetric matrix $W$.
Let
\begin{align*}
\hat w^*_{GMM,i}(\tilde\xi^{(b)})
=\frac{\tilde\xi^{(b)\prime}_2\tilde W(\tilde\xi^{(b)},\tilde\Sigma)e_ie_i'\tilde\xi_2^{(b)}}
  {{\tilde\xi_2^{(b)\prime}}\tilde W(\tilde\xi^{(b)},\tilde\Sigma){\tilde\xi_2^{(b)}}},
\end{align*}
where $e_i$ is the $i$th standard basis vector in $\mathbb{R}^k$, and let
\begin{align*}
\hat\beta_{GMM,\hat W}=\frac{\tilde\xi_2'\hat W\tilde\xi_1}
  {\tilde\xi_2'\hat W\tilde\xi_1}.
\end{align*}

The following theorem shows that $\hat\beta_{GMM,\hat W}$ and $\hat\beta_{RB}(\sqrt{T}\tilde\xi,\tilde\Sigma,\hat w^*_{GMM})$ are asymptotically equivalent in the strong instrument case.
For the case where $\tilde W(\tilde\xi^{(b)},\tilde \Sigma)=\hat W=Z'Z/T$, this gives asymptotic equivalence to 2SLS.

\begin{theorem}
Let $\tilde W(\tilde\xi^{(b)},\tilde \Sigma)$ and $\hat W$ be weighting matrices that converge in probability to the same positive definite matrix $W$, such that $\hat w^*_{GMM,i}$ defined above is uniformly bounded over $\tilde\xi^{(b)}$.  Then, under (\ref{rf_clt_eq}), (\ref{strong_iv_cond}) and (\ref{Sigma_consistency_cond}),
\begin{align*}
\sqrt{T}\left(\hat\beta_{RB}(\sqrt{T}\tilde\xi,\tilde\Sigma,\hat w^*_{GMM})-\hat\beta_{GMM,\hat W}\right)
\stackrel{p}{\to} 0.
\end{align*}
\end{theorem}
\begin{proof}
As with the normal case, define the oracle linear combination of 2SLS estimators
\begin{align*}
\hat\beta^o_{2SLS}=\sum_{i=1}^k w_i^*\frac{\tilde\xi_{1,i}}{\tilde\xi_{2,i}}
\end{align*}
where $w_i^*=\frac{{\pi^*}'W e_ie_i'{\pi^*}}{{\pi^*}'W {\pi^*}}$.
We have
$\sqrt{T}\left(\hat\beta_{RB}(\sqrt{T}\tilde\xi,\tilde\Sigma,\hat w^*_{GMM})-\hat\beta_{GMM,\hat W}\right)
=I+II+III$
where
$I\equiv \sqrt{T}(\hat\beta_{RB}(\sqrt{T}\tilde\xi,\tilde\Sigma,\hat w^*_{GMM})-\hat\beta_{RB}(\sqrt{T}\tilde\xi,\tilde\Sigma,w^*))$,
$II\equiv \sqrt{T}(\hat\beta_{RB}(\sqrt{T}\tilde\xi,\tilde\Sigma,w^*)-\hat\beta^o_{2SLS})$ and
$III\equiv \sqrt{T}(\hat\beta^o_{2SLS}-\hat\beta_{GMM,\hat W})$.

For the first term, note that
\begin{align*}
&I=\sqrt{T}\sum_{i=1}^k E_{\tilde\Sigma}\left[\left(\hat w^*_{GMM,i}(\tilde\xi^{(b)})-w_i^*\right)\hat\beta_U(\sqrt{T}\tilde\xi^{(a)}(i),\tilde\Sigma(i))\Big|\tilde\xi\right]  \\
&=\sqrt{T}\sum_{i=1}^k E_{\tilde\Sigma}\left[\left(\hat w^*_{GMM,i}(\tilde\xi^{(b)})-w_i^*\right)\left(\hat\beta_U(\sqrt{T}\tilde\xi^{(a)}(i),\tilde\Sigma(i))-\beta\right)\Big|\tilde\xi\right]
\end{align*}
where the last equality follows since 
$\sum_{i=1}^k \hat w^*_{GMM,i}(\tilde\xi^{(b)})=\sum_{i=1}^k w_i^*=1$
with probability one.
Thus, by H\"{o}lder's inequality,
\begin{align*}
|I|\le \sqrt{T}\sum_{i=1}^k \left(E_{\tilde\Sigma}\left[\left|\hat w^*_{GMM,i}(\tilde\xi^{(b)})-w_i^*\right|^q\Big|\tilde\xi\right]\right)^{1/q}
\left(E_{\tilde\Sigma}\left[\left|\hat\beta_U(\sqrt{T}\tilde\xi^{(a)}(i),\tilde\Sigma(i))-\beta\right|^p\Big|\tilde\xi\right]\right)^{1/p}
\end{align*}
for any $p$ and $q$ with $p,q>1$ and $1/p+1/q=1$ such that these conditional expectations exist.  Under (\ref{strong_iv_cond}),
$\hat w^*_{GMM,i}(\tilde\xi^{(b)})\stackrel{p}{\to}w_i^*$ so, since
$\hat w^*_{GMM,i}(\tilde\xi^{(b)})$ is uniformly bounded,
$E_{\tilde\Sigma}\left[\left|\hat w^*_{GMM,i}(\tilde\xi^{(b)})-w_i^*\right|^q\Big|\tilde\xi\right]$
will converge to zero for any $q$.
Thus, for this term, it suffices to bound
\begin{align*}
&\sqrt{T}\left(E_{\tilde\Sigma}\left[\left|\hat\beta_U(\sqrt{T}\tilde\xi^{(a)}(i),\tilde\Sigma(i))-\beta\right|^p\Big|\tilde\xi\right]\right)^{1/p}
=\sqrt{T}R\left(\sqrt{T}\tilde\xi(i),\tilde\Sigma(i),\beta,p\right)^{1/p}  \\
&=\sqrt{T}\tilde R\left(
\frac{\sqrt{T}\tilde\xi_2(i)}{\sqrt{\tilde\Sigma_{22}(i)}},
\frac{\sqrt{T}(\tilde\xi_1(i)-\beta\tilde\xi_2(i))}{\sqrt{\tilde\Sigma_{22}(i)}},
\left(\frac{\tilde\Sigma_{11}(i)}{\tilde\Sigma_{22}(i)}
-\frac{\tilde\Sigma_{12}(i)^2}{\tilde\Sigma_{22}(i)^2}\right)^{1/2},
\beta-\frac{\tilde\Sigma_{12}(i)}{\tilde\Sigma_{22}(i)},p
\right)^{1/p}
\end{align*}
for $R$ and $\tilde R$ as defined in Section \ref{auxiliary_lemmas_sec} below.
By Lemma \ref{tildeR_larget_bound_lemma} below, this is equal to $\sqrt{\tilde\Sigma_{22}(i)}/\tilde\xi_2(i)$ times a $\mathcal{O}_P(1)$ term for $p<2$.  Since $\sqrt{\tilde\Sigma_{22}(i)}/\tilde\xi_2(i)\stackrel{p}{\to}\sqrt{\Sigma_{22}^*(i)}/\pi_i^*$, it follows that the above display is also $\mathcal{O}_P(1)$.  Thus, $I\stackrel{p}{\to} 0$.

For the second term, we have
\begin{align*}
&II=\sqrt{T}\sum_{i=1}w_i^*\left(\hat\beta_U(\sqrt{T}\tilde\xi(i),\tilde\Sigma(i))-\frac{\tilde\xi_1(i)}{\tilde\xi_2(i)}\right)  \\
&=\sqrt{T}\sum_{i=1}w_i^*\left(\sqrt{T}\tilde\xi_2(i)\hat\tau(\sqrt{T}\tilde\xi_2(i),\tilde\Sigma_{22}(i))-1\right)
  \left(\frac{\tilde\xi_1(i)}{\tilde\xi_2(i)}-\frac{\tilde\Sigma_{12}(i)}{\tilde\Sigma_{22}(i)}\right).
\end{align*}
For each $i$, $\frac{\tilde\xi_1(i)}{\tilde\xi_2(i)}-\frac{\tilde\Sigma_{12}(i)}{\tilde\Sigma_{22}(i)}$ converges in probability to a finite constant and, by Section 2.3.4 of \citet{Small2010},
\begin{align*}
\sqrt{T}\left|\sqrt{T}\tilde\xi_2(i)\hat\tau(\sqrt{T}\tilde\xi_2(i),\tilde\Sigma_{22}(i))-1\right|
\le \sqrt{T}\frac{\tilde\Sigma_{22}(i)}{T\tilde\xi_2(i)^2}\stackrel{p}{\to} 0.
\end{align*}

The third term converges in probability to zero by standard arguments.  We have
\begin{align*}
III=\sqrt{T}
\sum_{i=1}^k\left(w_i^*-\frac{\tilde \xi_2' \hat W e_ie_i'\tilde\xi_2}{\tilde \xi_2' \hat W \tilde\xi_2}\right)\frac{\tilde\xi_{1,i}}{\tilde\xi_{2,i}}
=\sqrt{T}
\sum_{i=1}^k\left(w_i^*-\frac{\tilde \xi_2' \hat W e_ie_i'\tilde\xi_2}{\tilde \xi_2' \hat W \tilde\xi_2}\right)\left(\frac{\tilde\xi_{1,i}}{\tilde\xi_{2,i}}-\beta\right),
\end{align*}
where the last equality follows since $\sum_{i=1}^kw_i^*=\sum_{i=1}^k\frac{\tilde \xi_2' \hat W e_ie_i'\tilde\xi_2}{\tilde \xi_2' \hat W \tilde\xi_2}$ with probability one.  The result then follows from Slutsky's theorem.

\end{proof}

\subsection{Auxiliary Lemmas}\label{auxiliary_lemmas_sec}

For $p\ge 1$, $x\in\mathbb{R}^2$, $\Omega$ a $2\times 2$ matrix and $b\in\mathbb{R}$, let
\begin{align*}
R(x,\Omega,b,p)
=\int\int \left|\hat\beta_U(x+\Omega^{1/2}z,2\Omega)-b\right|^p
\phi(z_1)\phi(z_2)\, dz_1dz_2
\end{align*}
and let
\begin{align*}
\tilde R(t,c_1,c_2,c_3,p)
=\int\int
\left|
\hat\tau(t+z_2,2)(c_1+c_2z_1)
+\left[\hat\tau(t+z_2,2)t-1\right] c_3\right|^p
\phi(z_1)\phi(z_2)\, dz_1dz_2.
\end{align*}

\begin{lemma}\label{R_tildeR_lemma}
For $R$ and $\tilde R$ defined above,
\begin{align*}
R(x,\Omega,b,p)
=\tilde R\left(\frac{x_2}{\sqrt{\Omega_{22}}},\frac{x_1-b x_2}{\sqrt{\Omega_{22}}},
\left(\frac{\Omega_{11}}{\Omega_{22}}-\frac{\Omega_{12}^2}{\Omega_{22}^2}\right)^{1/2},
b-\frac{\Omega_{12}}{\Omega_{22}},p\right)
\end{align*}
\end{lemma}
\begin{proof}
Without loss of generality, we can let $\Omega^{1/2}$ be the upper diagonal square root matrix
\begin{align*}
\Omega^{1/2}=\left(\begin{array}{cc}
\left(\Omega_{11}-\frac{\Omega_{12}^2}{\Omega_{22}}\right)^{1/2}
& \frac{\Omega_{12}}{\sqrt{\Omega_{22}}}  \\
0 & \sqrt{\Omega_{22}}
\end{array}\right).
\end{align*}
Then
\begin{align*}
&\hat\beta_U(x+\Omega^{1/2} z, 2\Omega)  \\
&=\hat\tau(x_2+\sqrt{\Omega_{22}}z_2,2\Omega_{22})
\cdot \left(x_1+\left(\Omega_{11}-\frac{\Omega_{12}^2}{\Omega_{22}}\right)^{1/2}z_1+\frac{\Omega_{12}}{\sqrt{\Omega_{22}}} z_2
-\frac{\Omega_{12}}{\Omega_{22}}\left(x_2+\sqrt{\Omega_{22}}z_2\right)\right)
+\frac{\Omega_{12}}{\Omega_{22}}  \\
&=\frac{\hat\tau(x_2/\sqrt{\Omega_{22}}+z_2,2)}{\sqrt{\Omega_{22}}}
\cdot \left(x_1+\left(\Omega_{11}-\frac{\Omega_{12}^2}{\Omega_{22}}\right)^{1/2}z_1
-\frac{\Omega_{12}}{\Omega_{22}}x_2\right)
+\frac{\Omega_{12}}{\Omega_{22}}
\end{align*}
so that
\begin{align*}
&\hat\beta_U(x+\Omega^{1/2} z, 2\Omega)-b
=\frac{\hat\tau(x_2/\sqrt{\Omega_{22}}+z_2,2)}{\sqrt{\Omega_{22}}}
\cdot \left(x_1+\left(\Omega_{11}-\frac{\Omega_{12}^2}{\Omega_{22}}\right)^{1/2}z_1
-\frac{\Omega_{12}}{\Omega_{22}}x_2\right)
+\frac{\Omega_{12}}{\Omega_{22}}-b  \\
&=\frac{\hat\tau(x_2/\sqrt{\Omega_{22}}+z_2,2)}{\sqrt{\Omega_{22}}}
\cdot \left(x_1-x_2b+\left(\Omega_{11}-\frac{\Omega_{12}^2}{\Omega_{22}}\right)^{1/2}z_1
+\left(b-\frac{\Omega_{12}}{\Omega_{22}}\right)x_2\right)
+\frac{\Omega_{12}}{\Omega_{22}}-b  \\
&=\frac{\hat\tau(x_2/\sqrt{\Omega_{22}}+z_2,2)}{\sqrt{\Omega_{22}}}
\cdot \left(x_1-x_2b+\left(\Omega_{11}-\frac{\Omega_{12}^2}{\Omega_{22}}\right)^{1/2}z_1\right)
+\left(\frac{\hat\tau(x_2/\sqrt{\Omega_{22}}+z_2,2)}{\sqrt{\Omega_{22}}}x_2-1\right)
\left(b-\frac{\Omega_{12}}{\Omega_{22}}\right)
\end{align*}
and the result follows by plugging this in to the definition of $R$.

\end{proof}

We now give bounds on $R$ and $\tilde R$.
By the triangle inequality,
\begin{align}\label{tildeR_triang_ineq}
\tilde R(t,c_1,c_2,c_3,p)^{1/p}
&\le \left(\int\int
\hat\tau(t+z_2,2)^p|c_1+c_2z_1|^p
\phi(z_1)\phi(z_2)\, dz_1dz_2\right)^{1/p}  \nonumber  \\
&+c_3\left(\int\int
\left|\hat\tau(t+z_2,2)t-1\right|^p
\phi(z_1)\phi(z_2)\, dz_1dz_2\right)^{1/p}  \nonumber  \\
&=[C_1(t,p)\cdot C_2(c_1,c_2,p)]^{1/p}+c_3 C_3(t,p)^{1/p}
\end{align}
where
$C_1(t,p)=\int \hat\tau(t+z,2)^p\phi(z)\, dz$,
$C_2(c_1,c_2,p)=\int |c_1+c_2z|^p\phi(z)\, dz$
and
$C_3(t,p)=\int |\hat\tau(t+z,2) t - 1|^p\phi(z)\, dz$.
Note that, by the triangle inequality, for $t>0$,
\begin{align}\label{C1_bound_eq}
C_1(t,p)^{1/p}\le
\left(\int |\hat\tau(t+z,2) - 1/t|^p\phi(z)\, dz\right)^{1/p}+1/t
=(1/t)\left[C_3(t,p)^{1/p}+1\right].
\end{align}
Similarly,
\begin{align}\label{C3_bound_eq}
C_3(t,p)^{1/p}\le 1+t \left(\int \hat \tau(t+z,2)^p\phi(z)\, dz\right)^{1/p}
=1+t C_1(t,p)^{1/p}.
\end{align}

\begin{lemma}\label{tildeR_bound_lemma}
For $p<2$, $C_1(t,p)$ is bounded uniformly over $t$ on any compact set,
and $\tilde R(t,c_1,c_2,c_3,p)$ is bounded uniformly over $(t,c_1,c_2,c_3)$ in any compact set.
\end{lemma}
\begin{proof}
We have
\begin{align*}
&C_1(t,p)=\int \hat\tau(t+z,2)^p\phi(z)\, dz
=\int \left|\frac{1}{\sqrt{2}}\frac{1-\Phi\left((t+z)/\sqrt{2}\right)}{\phi\left((t+z)/\sqrt{2}\right)}\right|^p\phi(z)\, dz  \\
&\le 2^{-p/2}\int \frac{\phi(z)}{\phi((t+z)/\sqrt{2})^p}\, dz
\le K\int \exp\left(-\frac{1}{2}z^2+\frac{p}{4}(t+z)^2\right)\, dz
\end{align*}
for a constant $K$ that depends only on $p$.  This is bounded uniformly over $t$ in any compact set so long as $p/4<1/2$, giving the first result.  Boundedness of $\tilde R$ follows from this, (\ref{C3_bound_eq}) and boundedness of $C_2(c_1,c_2,p)$ over $c_1,c_2$ in any compact set.
\end{proof}

\begin{lemma}\label{tildeR_larget_bound_lemma}
For $p<2$,
$t \tilde R(t,c_1,c_2,c_3,p)^{1/p}$
is bounded uniformly over $t,c_1,c_2,c_3$ in any set such that $t$ is bounded from below away from zero and $c_1$, $c_2$ and $c_3$ are bounded.
\end{lemma}
\begin{proof}
By (\ref{tildeR_triang_ineq}) and (\ref{C1_bound_eq}), it suffices to bound
$t C_3(t,p)^{1/p}=t\left(\int |\hat\tau(t+z,2)t-1|^p\phi(z)\, dz\right)^{1/p}$.  Let $\varepsilon>0$ be a constant to be determined later.  We split the integral into the regions $t+z<\varepsilon t$ and $t+z\ge \varepsilon t$.  We have
\begin{align}\label{lower_int_eq}
&\int_{t+z<\varepsilon t} |\hat\tau(t+z,2)t-1|^p\phi(z)\, dz
=\int_{t+z<\varepsilon t} \left|t\frac{1-\Phi\left((t+z)/\sqrt{2}\right)}{\sqrt{2}\phi\left((t+z)/\sqrt{2}\right)}-1\right|^p\phi(z)\, dz  \nonumber  \\
&=\int_{t+z<\varepsilon t} \left|t\left[1-\Phi\left((t+z)/\sqrt{2}\right)\right]-\sqrt{2}\phi\left((t+z)/\sqrt{2}\right)\right|^p
\frac{\phi(z)}{\left[\sqrt{2}\phi\left((t+z)/\sqrt{2}\right)\right]^p}\, dz.
\end{align}
This is bounded by a constant times
\begin{align*}
&\max\{t,1\} \int_{t+z\le \varepsilon t}\exp\left(-\frac{1}{2}z^2+\frac{p}{4}\left(t+z\right)^2\right)\, dz  \\
&=\max\{t,1\} \int_{t+z\le \varepsilon t}\exp\left(-\frac{1}{2}z^2
  +\frac{p}{4}\left(z^2+2tz+t^2\right)\right)\, dz  \\
&=\max\{t,1\} \int_{t+z\le \varepsilon t}\exp\left(-\frac{1}{2}
\left(z^2(1-p/2)-t^2(p/2)-ptz
\right)\right)\, dz  \\
&=\max\{t,1\} \int_{t+z\le \varepsilon t}\exp\left(-\frac{1-p/2}{2}
\left(z^2-t^2\frac{p}{2-p}-2\frac{p}{2-p}tz
\right)\right)\, dz  \\
&=\max\{t,1\} \int_{t+z\le \varepsilon t}\exp\left(-\frac{1-p/2}{2}
\left(\left(z-\frac{p}{2-p}t\right)^2
-\left(\frac{p}{2-p}\right)^2t^2
-\frac{p}{2-p}t^2
\right)\right)\, dz  \\
&=\max\{t,1\} \exp\left(\frac{1-p/2}{2}\left(\frac{p}{2-p}+\left(\frac{p}{2-p}\right)^2\right)t^2\right)
\int_{t+z\le \varepsilon t}\exp\left(-\frac{1-p/2}{2}
\left(z-\frac{p}{2-p}t\right)^2\right)\, dz.
\end{align*}
We have
\begin{align*}
&\int_{t+z\le \varepsilon t}\exp\left(-\frac{1-p/2}{2}
\left(z-\frac{p}{2-p}t\right)^2\right)\, dz  \\
&=\int_{z-t p/(2-p)\le (\varepsilon-1-p/(2-p)) t}\exp\left(-\frac{1-p/2}{2}
\left(z-\frac{p}{2-p}t\right)^2\right)\, dz  \\
&=\int_{u\le (\varepsilon-1-p/(2-p)) t}\exp\left(-\frac{1-p/2}{2}
u^2\right)\, dz,
\end{align*}
which is bounded by a constant times
\begin{align*}
\Phi\left(
t(\varepsilon-1-p/(2-p))\sqrt{1-p/2}
\right).
\end{align*}
For $t(\varepsilon-1-p/(2-p))<0$, this is bounded by a constant times
$\exp\left(-\frac{1-p/2}{2}t^2(1+p/(2-p)-\varepsilon)^2\right)$.
Thus, (\ref{lower_int_eq}) is bounded uniformly over $t>0$ by a constant times $\exp(-\eta t^2)$ for some $\eta>0$ so long as
\begin{align*}
\left(1+\frac{p}{2-p}-\varepsilon\right)^2>\frac{p}{2-p}+\left(\frac{p}{2-p}\right)^2
=\frac{p}{2-p}\left(1+\frac{p}{2-p}\right)
\end{align*}
which can be ensured by choosing $\varepsilon>0$ small enough so long as $p<2$.  Thus, $\varepsilon>0$ can be chosen so that (\ref{lower_int_eq}) is bounded uniformly over $t$ when scaled by $t^p$.

For the integral over $t+z>\varepsilon t$, we have, by (1.1) in \citet{baricz_mills_2008},
\begin{align*}
&\int_{t+z\ge \varepsilon t}\left|t\hat\tau(t+z,2)-1\right|^p\phi(z)\, dz
=t^p\int_{t+z\ge \varepsilon t}\left|\hat\tau(t+z,2)-\frac{1}{t}\right|^p\phi(z)\, dz  \\
&\le t^p\int_{t+z\ge \varepsilon t}\left|\frac{1}{t+z}-\frac{1}{t}\right|^p\phi(z)\, dz
+t^p\int_{t+z\ge \varepsilon t}\left|\frac{1}{(t+z)+2/(t+z)}-\frac{1}{t}\right|^p\phi(z)\, dz.
\end{align*}
The first term is
\begin{align*}
t^p\int_{t+z\ge \varepsilon t}\left|\frac{z}{(t+z)t}\right|^p\phi(z)\, dz
\le \frac{1}{t^p}\int\left|\frac{z}{\varepsilon}\right|^p\phi(z)\, dz.
\end{align*}
The second term is
\begin{align*}
t^p\int_{t+z\ge \varepsilon t}\left|\frac{-z-2/(t+z)}{[(t+z)+2/(t+z)]t}\right|^p\phi(z)\, dz
\le \frac{1}{t^p}\int \left|\frac{|z|+|2/\varepsilon t|}{\varepsilon}\right|^p\phi(z)\, dz.
\end{align*}
Both are bounded uniformly when scaled by $t^p$ over any set with $t$ bounded from below away from zero.
\end{proof}

\section{Relation to \citet{HiranoPoter2015}}\label{hp_append}

\citet{HiranoPoter2015} give a negative result establishing
the impossibility of unbiased, quantile unbiased, or translation equivariant
estimation in a wide variety of models with singularities, including
many linear IV models. On initial inspection our derivation of an
unbiased estimator for $\beta$ may appear to contradict the results
of \citeauthor{HiranoPoter2015}.  In fact, however,
one of the key assumptions of \citet{HiranoPoter2015} no longer applies 
once we assume that the sign of the first stage is known.

Again consider the linear IV model with a single instrument,
where for simplicity we let $\sigma_{1}^{2}=\sigma_{2}^{2}=1,$
$\sigma_{12}=0$. To discuss the results of \citet{HiranoPoter2015}, it will be helpful
to parameterize the model in terms of the reduced-form parameters $(\psi,\pi)=(\pi\beta,\pi).$
For $\phi$ again the standard normal density, the density of $\xi$ is 
\[
f\left(\xi;\psi,\pi\right)=\mbox{\ensuremath{\phi\left(\xi_{1}-\psi\right)\phi\left(\xi_{2}-\pi\right)}}.
\]
Fix some value $\psi^*$.  For any $\pi\neq 0$ we can
define $\beta(\psi,\pi)=\frac{\psi}{\pi}$.  If we consider any sequence $\{\pi_j\}_{j=1}^\infty$
approaching zero from the right, then $\beta(\psi^*,\pi_j)\to\infty$ if $\psi^*>0$ and 
$\beta(\psi^*,\pi_j)\to-\infty$ if $\psi^*<0$.  Thus we can see that $\beta$ plays the role of the function $\kappa$  in
\citet{HiranoPoter2015} equation (2.1).

\citet{HiranoPoter2015} show that if there exists some finite collection
of parameter values $(\psi_{l,d},\pi_{l,d})$ in the parameter space and non-negative
constants $c_{l,d}$ such that their Assumption 2.4, 
\[
f\left(\xi;\psi^*,0\right)\le\sum_{l=1}^{s}c_{l,d}f\left(\xi;\psi_{l,d},\pi_{l,d}\right)\,\forall\xi,
\]
holds, then (since one can easily verify their Assumption 2.3 in the present context)
there can exist no unbiased estimator of $\beta.$

This
dominance condition fails in the linear IV model with a sign restriction. 
For any $(\psi_{l,d},\pi_{l,d})$
in the parameter space, we have by definition that $\pi_{l,d}>0$.
For any such $\pi_{l,d}$, however, if we fix $\xi_{1}$ and take
$\xi_{2}\to-\infty$, 
\[
\lim_{\xi_{2}\to-\infty}\frac{\phi\left(\xi_{2}-\pi_{l,d}\right)}{\phi\left(\xi_{2}\right)}=\lim_{\xi_{2}\to-\infty}\exp\left(-\frac{1}{2}\left(\xi_{2}-\pi_{l,d}\right)^{2}+\frac{1}{2}\xi_{2}^{2}\right)=\lim_{\xi_{2}\to-\infty}\exp\left(\xi_{2}\pi_{l,d}-\frac{1}{2}\pi_{l,d}^{2}\right)=0.
\]

Thus, $\lim_{\xi_{2}\to-\infty}\frac{f\left(\xi;\psi_{l,d},\pi_{l,d}\right)}{f\left(\xi;\psi^*,0\right)}=0$,
and for any fixed $\xi_{1},$ $\left\{ c_{l,d}\right\} _{l=1}^{s}$
and $\left\{(\psi_{l,d},\pi_{l,d})\right\} _{l=1}^{s}$ there exists a $\xi_{2}^{*}$
such that $\xi_{2}<\xi_{2}^{*}$ implies 
\[
f\left(\xi;\psi^{*},0\right)>\sum_{l=1}^{s}c_{l,d}f\left(\xi;\psi_{l,d},\pi_{l,d}\right).
\]
Thus, Assumption 2.4 in \citet{HiranoPoter2015} fails in this model, allowing
the possibility of an unbiased estimator. Note, however, that if we did not impose $\pi>0$ then we would
satisfy Assumption 2.4, so unbiased estimation of $\beta$
would again be impossible. Thus, the sign restriction on $\pi$
plays a central role in the construction of the unbiased estimator
$\hat{\beta}_{U}.$

\section{Lower Bound on Risk of Unbiased Estimators}\label{risk_lower_bound_sec}

This appendix gives a lower bound on the attainable risk at a given $\pi,\beta$ for an estimator that is unbiased for $\beta$ for all $\pi,\beta$ with $\pi$ in the positive orthant.  The bound is given by the risk in the submodel where $\pi/\|\pi\|$ (the direction of $\pi$) is known.  While the bound cannot, in general, be obtained, we discuss some situations where it can, which include certain values of $\pi$ in the case where $\xi$ comes from a model with homoskedastic errors.

\begin{theorem}\label{risk_bound_thm}
Let $\mathcal{U}$ be the set of estimators for $\beta$ that are unbiased for all $\pi\in(0,\infty)^k$, $\beta\in\mathbb{R}$.
For any $\pi^*\in (0,\infty)^k$, $\beta^*\in\mathbb{R}$ and any convex loss function $\ell$,
\begin{align*}
E_{\pi^*,\beta^*} \ell(\hat\beta_U(\xi^*(\pi^*),\Sigma^*(\pi^*))-\beta^*)
\le \inf_{\hat\beta\in \mathcal{U}} E_{\pi^*,\beta^*} \ell(\hat\beta(\xi,\Sigma)-\beta^*)
\end{align*}
where
$\xi^*(\pi^*)=\left[\left(I_2\otimes \pi^*\right)'\Sigma^{-1}\left(I_2\otimes \pi^*\right)\right]^{-1}
  \left(I_2\otimes \pi^*\right)'\Sigma^{-1}\xi$
and
$\Sigma^*(\pi^*)=\left[\left(I_2\otimes \pi^*\right)'\Sigma^{-1}\left(I_2\otimes \pi^*\right)\right]^{-1}$.

\end{theorem}
\begin{proof}
Consider the submodel with $\pi$ restricted to $\Pi^*=\{\pi^* t|t\in(0,\infty)\}$.  Then $\xi^*(\pi^*)$ is sufficient for $(t,\beta)$ in this submodel, and satisfies
$\xi^*(\pi^*)\sim N((\beta t,t)',\Sigma^*(\pi^*))$
in this submodel.
To see this, note that, for $t,\beta$ in this submodel, $\xi$ follows the generalized least squares regression model $\xi=(I_2\otimes \pi^*)(\beta t,t)'+\varepsilon$ where $\varepsilon\sim N(0,\Sigma)$, and $\xi^*(\pi^*)$ is the generalized least squares estimator of $(\beta t,t)'$.

Let $\tilde\beta(\xi(\pi^*),\Sigma(\pi^*))$ be a (possibly randomized) estimator based on $\xi(\pi^*)$ that is unbiased in the submodel where $\pi\in\Pi^*$.
By completeness of the submodel, $E\left[\tilde\beta(\xi(\pi^*),\Sigma(\pi^*))|\xi^*(\pi^*)\right]=\hat\beta_U(\xi(\pi^*),\Sigma(\pi^*))$.
By Jensen's inequality, therefore,
\begin{align*}
E_{\pi^*,\beta^*}\ell\left(\tilde\beta(\xi(\pi^*),\Sigma(\pi^*))-\beta\right)
\ge  E_{\pi^*,\beta^*}\ell\left(E\left[\tilde\beta(\xi(\pi^*),\Sigma(\pi^*))|\xi^*(\beta)\right]-\beta\right)
\end{align*}
(this is just Rao-Blackwell applied to the submodel with the loss function $\ell$).
By sufficiency, the set of risk functions for randomized unbiased estimators based on $\xi(\pi^*)$ in the submodel is the same as the set of risk functions for randomized unbiased estimators based on $\xi$ in the submodel.  This gives the result with $\mathcal{U}$ replaced by the set of estimators that are unbiased in the submodel, which implies the result as stated, since the set of estimator which are unbiased in the full model is a subset of those which are unbiased in the submodel.
\end{proof}

Theorem \ref{risk_bound_thm} continues to hold in the case where the lower bound is infinite: in this case, the risk of any unbiased estimator must be infinite at $\beta^*,\pi^*$.  By Theorem \ref{betau_second_moment_thm}, the lower bound is infinite for squared error loss $\ell(t)=t^2$ for any $\pi^*,\beta^*$.  Thus, unbiased estimators must have infinite variance even in models with multiple instruments.

While in general Theorem \ref{risk_bound_thm} gives only a lower bound on the risk of unbiased estimators, the bound can be achieved in certain situations.  
A case of particular interest arises in models with homoskedastic reduced form errors that are independent across observations.  In such cases $Var\left(\left(U',V'\right)'\right)=Var\left((U_1,V_1)'\right)\otimes I_T$, where $I_T$ is the $T\times T$ identity matrix, so that the definition of $\Sigma$ in (\ref{Sigma_def}) gives
$\Sigma
=Var\left((U_1,V_1)'\right)\otimes \left(Z'Z\right)^{-1}$.
Thus, in models with independent homoskedastic errors we have $\Sigma=Q_{UV}\otimes Q_{Z}$ for a $2\times 2$ matrix $Q_{UV}$ and a $k\times k$ matrix $Q_Z$.

\begin{theorem}\label{risk_attain_thm}
Suppose that
$\left[\left(I_2\otimes \pi^*\right)'\Sigma^{-1}\left(I_2\otimes \pi^*\right)\right]^{-1}
  \left(I_2\otimes \pi^*\right)'\Sigma^{-1}
=(I_2\otimes a(\pi^*)')$
for some $a(\pi^*)\in\mathbb{R}^k$.  Then $\hat\beta_U(\xi^*(\pi^*),\Sigma(\pi^*))$ defined in Theorem \ref{risk_bound_thm} is unbiased at any $\pi,\beta$ such that $a(\pi^*)'\pi>0$.  In particular, if $a(\pi^*)\in (0,\infty)^k$, then
$\hat\beta_U(\xi^*(\pi^*),\Sigma(\pi^*))\in \mathcal{U}$ and the risk bound is attained.
Specializing to the case where $\Sigma=Q_{UV}\otimes Q_{Z}$ for a $2\times 2$ matrix $Q_{UV}$ and a $k\times k$ matrix $Q_Z$, the above conditions hold with
$a(\pi^*)'=\pi^{*\prime}Q_Z^{-1}/(\pi^{*\prime}Q_Z^{-1}\pi^*)$, and the bound is achieved if
$Q_Z^{-1}\pi^*\in (0,\infty)^k$.
\end{theorem}
\begin{proof}
For the first claim, note that under these assumptions
$\xi^*(\pi^*)=(a(\pi^*)'\xi_1,a(\pi^*)'\xi_2)'$ is $N((a(\pi^*)'\pi \beta,a(\pi^*)'\pi)',\Sigma^*(\pi))$ distributed under $\pi,\beta$, so $\hat\beta_U(\xi^*(\pi^*),\Sigma(\pi^*))$ is unbiased at $\pi,\beta$ by Theorem \ref{Thm: Unbiased Estimator, Just-identified Case}.  For the case where $\Sigma=Q_{UV}\otimes Q_{Z}$, the result follows by properties of the Kronecker product:
\begin{align*}
&\left[\left(I_2\otimes \pi^*\right)'\left(Q_{UV}\otimes Q_{Z}\right)^{-1}\left(I_2\otimes \pi^*\right)\right]^{-1}
  \left(I_2\otimes \pi^*\right)'\left(Q_{UV}\otimes Q_{Z}\right)^{-1}  \\
&=\left[Q_{UV}^{-1}\otimes \pi^{*\prime}Q_Z^{-1}\pi^*\right]^{-1}
    \left(Q_{UV}^{-1}\otimes \pi^{*\prime} Q_{Z}^{-1}\right)
=I_2\otimes \left[\pi^{*\prime} Q_{Z}^{-1}/\left(\pi^{*\prime}Q_Z^{-1}\pi^*\right)\right].
\end{align*}
\end{proof}

The special form of the sufficient statistic in the homoskedastic case derives from the form of the optimal estimator in the restricted seemingly unrelated regression (SUR) model.  The submodel for the direction $\pi^*$ is given by the SUR model
\begin{align*}
\left(\begin{array}{c} Y  \\
X
\end{array}\right)
=\left(\begin{array}{cc} Z\pi^* & 0  \\
0 & Z\pi^*
\end{array}\right)
\left(\begin{array}{c}\beta t  \\
t
\end{array}\right)
+\left(\begin{array}{c} U  \\
V \end{array}\right).
\end{align*}
Considering this as a SUR model with regressors $Z\pi^*$ in both equations, the optimal estimator of $(\beta t,t)'$ simply stacks the OLS estimator for the two equations, since the regressors $Z\pi^*$ are the same and the parameter space for $(\beta t, t)$ is unrestricted.
Note also that, in the homoskedastic case (with $Q_Z=(Z'Z)^{-1}$),
$\xi_1^*(\pi^*)$ and $\xi_2^*(\pi^*)$ are proportional to $\pi^{*\prime}Z'Z \xi_1$ and $\pi^{*\prime}Z'Z \xi_2$, which are the numerator and denominator of the 2SLS estimator with $\xi_2$ replaced by $\pi^{*}$ in the first part of the quadratic form.

Thus, for certain parameter values $\pi^*$ in the homoskedastic case, the risk bound in Theorem \ref{risk_bound_thm} is obtained.  In such cases, the estimator that obtains the bound is unique, and depends on $\pi^*$ itself (for the absolute value loss function, which is not strictly concave, uniqueness is shown in Section \ref{unique_unbiased_sec} below).  Thus, in contrast to settings such as linear regression, where a single estimator minimizes the risk over unbiased estimators simultaneously for all parameter values, no uniform minimum risk unbiased estimator will exist.  The reason for this is clear: knowledge of the direction of $\pi=\pi^*$ helps with estimation of $\beta$, even if one imposes unbiasedness for all $\pi$.

It is interesting to note precisely how the parameter space over which the estimator in the risk bound is unbiased depends on $\pi^*$.  Suppose one wants an estimator that minimizes the risk at $\pi^*$ while still remaining unbiased in a small neighborhood of $\pi^*$.  In the homoskedastic case, this can always be done so long as $\pi^*\in (0,\infty)^k$, since $\pi^{*\prime}Q_Z^{-1}\pi>0$ for $\pi$ close enough to $\pi^*$.  Where one can expand this neighborhood while maintaining unbiasedness will depend on $\pi^*$ and $Q_Z$.  In the case where $\pi^{*\prime}Q_Z^{-1}$ is in the positive orthant, the assumption $\pi\in (0,\infty)^k$ is enough to ensure that this estimator is unbiased at $\pi$.  However, if $\pi^{*\prime}Q_Z^{-1}$ is not in the positive orthant, there is a tradeoff between precision at $\pi^*$ and the range of $\pi\in (0,\infty)^k$ over which unbiasedness can be maintained.

Put another way, in the homoskedastic case, for any $\pi^*\in\mathbb{R}^k\backslash\{0\}$, minimizing the risk of an estimator of $\beta$ subject to the restriction that the estimator is unbiased in a neighborhood of $\pi^*$ leads to an estimator that does not depend on this neighborhood, so long as the neighborhood is small enough (this is true even if the restriction $\pi^*\in(0,\infty)^k$ does not hold).  The resulting estimator depends on $\pi^*$, and is unbiased at $\pi$ iff $\pi^*Q_Z^{-1}\pi>0$.

\subsection{Uniqueness of the Minimum Risk Unbiased Estimator under Absolute Value Loss}\label{unique_unbiased_sec}

In the discussion above, we used the result that the minimum risk unbiased estimator in the submodel with $\pi/\|\pi\|$ known is unique for absolute value loss.  Because the absolute value loss function is not strictly concave, this result does not, to our knowledge, follow immediately from results in the literature.
We therefore provide a statement and proof here.
In the following theorem, we consider a general setup where a random variable $\xi$ is observed, which follows a distribution $P_\mu$ for some $\mu\in M$.  The family of distributions $\{P_\mu|\mu\in M\}$ need not be a multivariate normal family, as in the rest of this paper.

\begin{theorem}
Let $\hat\theta=\hat\theta(\xi)$ be an unbiased estimator of $\theta=\theta(\mu)$ where $\mu\in M$ for some parameter space $M$ and $\Theta=\{\theta|\theta(\mu)=\theta\text{ some } \mu\in M\}\subseteq\mathbb{R}$, and where $\xi$ has the same support for all $\mu\in M$.  Let $\tilde\theta(\xi,U)$ be another unbiased estimator, based on $(\xi,U)$ where $\xi$ and $U$ are independent and $\hat\theta(\xi)=E_\mu[\tilde\theta(\xi,U)|\xi]=\int \tilde\theta(\xi,U)\, dQ(U)$ where $Q$ denotes the probability measure of $U$, which is assumed not to depend on $\mu$.  Suppose that $\hat\theta(\xi)$ and $\tilde\theta(\xi,U)$ have the same risk under absolute value loss:
\begin{align*}
E_\mu |\tilde\theta(\xi,U)-\theta(\mu)|=E_\mu |\hat\theta(\xi)-\theta(\mu)|
\text{ for all }\mu\in M.
\end{align*}
Then $\tilde\theta(\xi,U)=\hat\theta(\xi)$ for almost every $\xi$ with $\hat\theta(\xi)\in\Theta$.
\end{theorem}
\begin{proof}
The display can be written as
\begin{align*}
E_\mu\left\{ E_\mu\left[ |\tilde\theta(\xi,U)-\theta(\mu)|\Big|\xi\right]-|\hat\theta(\xi)-\theta(\mu)|\right\}=0
\text{ for all }\mu\in M.
\end{align*}
By Jensen's inequality, the term inside the outer expectation is nonnegative for $\mu$-almost every $\xi$.  Thus, the equality implies that this term is zero for $\mu$-almost every $\xi$ (since $EX=0$ implies $X=0$ a.e. for any nonnegative random variable $X$).  This gives, noting that
$\int |\tilde\theta(\xi,U)-\theta(\mu)|\, dQ(U)=E_\mu\left[ |\tilde\theta(\xi,U)-\theta(\mu)|\Big|\xi\right]$,
\begin{align*}
\int |\tilde\theta(\xi,U)-\theta(\mu)|\, dQ(U)
=|\hat\theta(\xi)-\theta(\mu)|
\text{ for $\mu$-almost every }\xi
\text{ and all }\mu\in M.
\end{align*}
Since the support of $\xi$ is the same under all $\mu\in M$, the above statement gives
\begin{align*}
\int |\tilde\theta(\xi,U)-\theta|\, dQ(U)
=|\hat\theta(\xi)-\theta|
\text{ for almost every }\xi
\text{ and all }\theta\in\Theta.
\end{align*}
Note that, for any random variable $X$, $E|X|=|EX|$ implies that either $X\ge 0$ a.e. or $X\le 0$ a.e.  Applying this to the above display, it follows that for all $\theta\in\Theta$ and almost every $\xi$, either
$\tilde\theta(\xi,U)\le \theta$ a.e. $U$ or $\tilde\theta(\xi,U)\ge \theta$ a.e. $U$.  In particular, whenever $\hat\theta(\xi)\in\Theta$, either $\tilde\theta(\xi,U)\le \hat\theta(\xi)$ a.e. $U$ or $\tilde\theta(\xi,U)\ge \hat\theta(\xi)$ a.e. $U$.  In either case, the condition $\int \tilde\theta(\xi,U)\, dQ(U)=\hat\theta(\xi)$ implies that $\tilde\theta(\xi,U)= \hat\theta(\xi)$ a.e. $U$.
It follows that, for almost every $\xi$ such that $\hat\theta(\xi)\in\Theta$, we have $\tilde\theta(\xi,U)= \theta(\xi)$ a.e. $U$, as claimed.
\end{proof}

Thus, if $\hat\theta(\xi)\in\Theta$ with probability one, we will have $\tilde\theta(\xi,U)=\hat\theta(\xi)$ a.e. $(\xi,U)$.  However, if $\hat\theta(\xi)$ can take values outside $\Theta$ this will not necessarily be the case.  For example, in the single instrument case of our setup, if we restrict our parameter space to $(\pi,\beta)\in (0,\infty)\times [c,\infty)$ for some constant $c$, then forming a new estimator by adding or subtracting $1$ from $\hat\beta_U$ with equal probability independently of $\xi$ whenever $\hat\beta_U\le c-1$ gives an unbiased estimator with identical absolute value risk.

In our case, letting $\xi(\pi^*)$ be as Theorem \ref{risk_bound_thm}, the support of $\xi(\pi^*)$ is the same under $\pi^*t$, $\beta$ for any $t\in(0,\infty)$ and $\beta\in\mathbb{R}$.  If $\tilde\beta(\xi(\pi^*),U)$ is unbiased in this restricted parameter space, we must have, letting $\hat\beta_U(\xi^*(\pi),\Sigma^*(\pi))$ be the unbiased nonrandomized estimator in the submodel, $E[\tilde\beta(\xi(\pi^*),U)|\xi(\pi^*)]=\hat\beta_U(\xi(\pi^*),\Sigma^*(\pi))$ by completeness for any random variable $U$ with a distribution that does not depend on $(t,\beta)$.  Since $\hat\beta_U(\xi(\pi^*),\Sigma^*(\pi))\in\mathbb{R}$ with probability one, it follows that if $\tilde\beta(\xi(\pi^*),U)$ has the same risk as $\hat\beta_U(\xi(\pi^*),\Sigma^*(\pi))$ then $\tilde\beta(\xi(\pi^*),U)=\hat\beta_U(\xi(\pi^*),\Sigma^*(\pi))$ with probability one, so long as we impose that $\tilde\beta(\xi(\pi^*),U)$ is unbiased for all $t\in(1-\varepsilon,1+\varepsilon)$ and $\beta\in\mathbb{R}$.

\section{Reduction of the Parameter Space by Equivariance}\label{equivariance_append}
In the appendix, we discuss how we can reduce the dimension of the parameter space using an equivariance argument.  We first consider the just-identified
case and then note how such arguments may be extended to the over-identified case under the additional assumption of homoskedasticity.

\subsection{Just-Identified Model}
For comparisons between $\left(\hat{\beta}_{U},\hat{\beta}_{2SLS},\hat{\beta}_{FULL}\right)$
in the just-identified case, it suffices to consider a two-dimensional
parameter space. To see that this is the case let $\theta=\left(\beta,\pi,\sigma_{1}^{2},\sigma_{12},\sigma_{2}^{2}\right)$
be the vector of model parameters and let $g_{A}$, for $A=\left[\begin{array}{cc}
a_{1} & a_{2}\\
0 & a_{3}
\end{array}\right],$ $a_{1}\neq0,$ $a_{3}>0$, be the transformation 
\[
g_{A}\xi=\tilde{\xi}=A\left(\begin{array}{c}
\xi_{1}\\
\xi_{2}
\end{array}\right)=\left(\begin{array}{c}
a_{1}\xi_{1}+a_{2}\xi_{2}\\
a_{3}\xi_{2}
\end{array}\right)
\]
which leads to $\tilde\xi$ being distributed according to the parameters
\[
\tilde{\theta}=\left(\tilde{\beta},\tilde{\pi},\tilde{\sigma}_{1}^{2},\tilde{\sigma}_{12},\tilde{\sigma}_{2}^{2}\right)
\]
where
\[
\tilde{\beta}=\frac{\left(a_{1}\beta+a_{2}\right)}{a_{3}}
\]
\[
\tilde{\pi}=a_{3}\pi
\]
\[
\tilde{\sigma}_{1}^{2}=a_{1}^{2}\sigma_{1}^{2}+a_{1}a_{2}\sigma_{12}+a_{2}^{2}\sigma_{2}^{2}
\]
\[
\tilde{\sigma}_{12}=a_{1}a_{3}\sigma_{12}+a_{2}a_{3}\sigma_{2}^{2}
\]
and 
\[
\tilde{\sigma}_{2}^{2}=a_{3}^{2}\sigma_{2}^{2}.
\]

Define $\mathcal{G}$ as the set of all transformations $g_{A}$ of
the form above. Note that the sign restriction on $\pi$ is preserved
under $g_{A}\in\mathcal{G}$, and that for each $g_{A}$, there exists
another transformation $g_{A}^{-1}\in\mathcal{G}$ such that $g_{A}g_{A}^{-1}$
is the identity transformation. We can see that the model (\ref{eq: Sufficient Statistics Def}) is invariant under the transformation $g_{A}$.
Note further that the estimators $\hat{\beta}_{U},$ $\hat{\beta}_{2SLS},$
and $\hat{\beta}_{FULL}$ are all equivariant under $g_{A}$, in the
sense that 
\[
\hat{\beta}\left(g_{A}\xi\right)=\frac{a_{1}\hat{\beta}\left(\xi\right)+a_{2}}{a_{3}}.
\]
Thus, for any properties of these estimators (e.g. relative mean and
median bias, relative dispersion) which are preserved under the transformations
$g_{A}$, it suffices to study these properties on the reduced parameter
space obtained by equivariance. By choosing $A$ appropriately, we
can always obtain 
\[
\left(\begin{array}{c}
\tilde{\xi}_{1}\\
\tilde{\xi}_{2}
\end{array}\right)\sim N\left(\left(\begin{array}{c}
0\\
\tilde{\pi}
\end{array}\right),\left(\begin{array}{cc}
1 & \tilde{\sigma}_{12}\\
\sigma_{12} & 1
\end{array}\right)\right)
\]
for $\tilde{\pi}>0$, $\sigma_{12}\ge0$ and thus reduce to a two-dimensional
parameter $\left(\pi,\sigma_{12}\right)$ with $\sigma_{12}\in[0,1)$,
$\pi>0$. 

\subsection{Over-Identified Model under Homoskedasticity}
As noted in Appendix \ref{risk_lower_bound_sec}, under the assumption of iid homoskedastic errors $\Sigma$ is of the form $\Sigma=Q_{UV}\otimes Q_{Z}$ for matrix $Q_{UV}=Var((U_1,V_1)')$ and  $Q_Z=(Z'Z)^{-1}$.  If we let $\sigma_U^2=Var(U_1)$, $\sigma_V^2=Var(V_1)$, and $\sigma_{UV}=Cov(U_1,V_1)$, then using an equivariance argument as above we can eliminate the parameters $\sigma_U^2$, $\sigma_V^2$, and $\beta$ for the purposes of comparing $\hat\beta_{2SLS},$ $\hat\beta_{FULL},$ and the unbiased estimators.  In particular, define $\theta=\left(\beta,\pi,\sigma_{U}^{2},\sigma_{UV},\sigma_{V}^{2},Q_Z\right)$
and again let $A=\left[\begin{array}{cc}
a_{1} & a_{2}\\
0 & a_{3}
\end{array}\right],$ $a_{1}\neq0,$ $a_{3}>0$ and consider the transformation
\[
g_{A}\xi=\tilde{\xi}=\left(A\otimes I_k\right)\left(\begin{array}{c}
\xi_{1}\\
\xi_{2}
\end{array}\right)=\left(\begin{array}{c}
a_{1}\xi_{1}+a_{2}\xi_{2}\\
a_{3}\xi_{2}
\end{array}\right)
\]
which leads to $\tilde\xi$ being distributed according to the parameters
\[
\tilde{\theta}=\left(\tilde{\beta},\tilde{\pi},\tilde{\sigma}_{U}^{2},\tilde{\sigma}_{UV},\tilde{\sigma}_{V}^{2},\tilde Q_Z\right)
\]
where
\[
\tilde{\beta}=\frac{\left(a_{1}\beta+a_{2}\right)}{a_{3}}
\]
\[
\tilde{\pi}=a_{3}\pi
\]
\[
\tilde{\sigma}_{U}^{2}=a_{1}^{2}\sigma_{U}^{2}+a_{1}a_{2}\sigma_{UV}+a_{2}^{2}\sigma_{V}^{2}
\]
\[
\tilde{\sigma}_{UV}=a_{1}a_{3}\sigma_{UV}+a_{2}a_{3}\sigma_{V}^{2}
\]
\[
\tilde{\sigma}_{V}^{2}=a_{3}^{2}\sigma_{V}^{2}
\]
and
\[
\tilde{Q_Z}=Q_Z.
\]

Note that this transformation changes neither the direction of the fist stage, $\pi/\|\pi\|$, nor $Q_Z.$  If we again define $\mathcal{G}$ to be the class of such transformations, we again see that the model is invariant under transformations $g_A\in\mathcal{G},$ and that the estimators for $\beta$ we consider are equivariant under these transformations.  Thus, since relative bias and MAD across estimators are preserved under these transformations, we can again study these properties on the reduced parameter space obtained by equivariance.  In particular, by choosing $A$ appropriately we can set $\tilde\sigma_U^2=\tilde\sigma_V^2=1$ and $\tilde\beta=0$, so the remaining free parameters are $\tilde\pi$, $\tilde\sigma_{UV},$ and $\tilde Q_Z$.

\section{Additional Simulation Results in Just-Identified Case}\label{just-id sim append}
This appendix gives further results for our simulations in the 
just-identified case.  We first report median bias comparisons for the estimators $\hat\beta_U$, $\hat\beta_{2SLS}$, and $\hat\beta_{FULL}$, and then report further dispersion and absolute deviation simulation results to complement
those in Section \ref{sec: absolute deviation sims} of the paper.

\subsection{Median Bias}\label{median bias append}
Figure \ref{fig Single Medan Instrument Bias} plots the median bias of the single-instrument IV estimators against the mean of the first stage F statistic.  In all calibrations considered the unbiased estimator has a smaller median bias than 2SLS when the first stage is very small and a larger median bias for larger values of the first stage.  By contrast the median bias of Fuller is larger than that of both the unbiased and 2SLS estimators, though its median bias is quite close to that of the unbiased estimator once the mean of the first stage F statistic exceeds 10.

\begin{figure}
\includegraphics[scale=0.65]{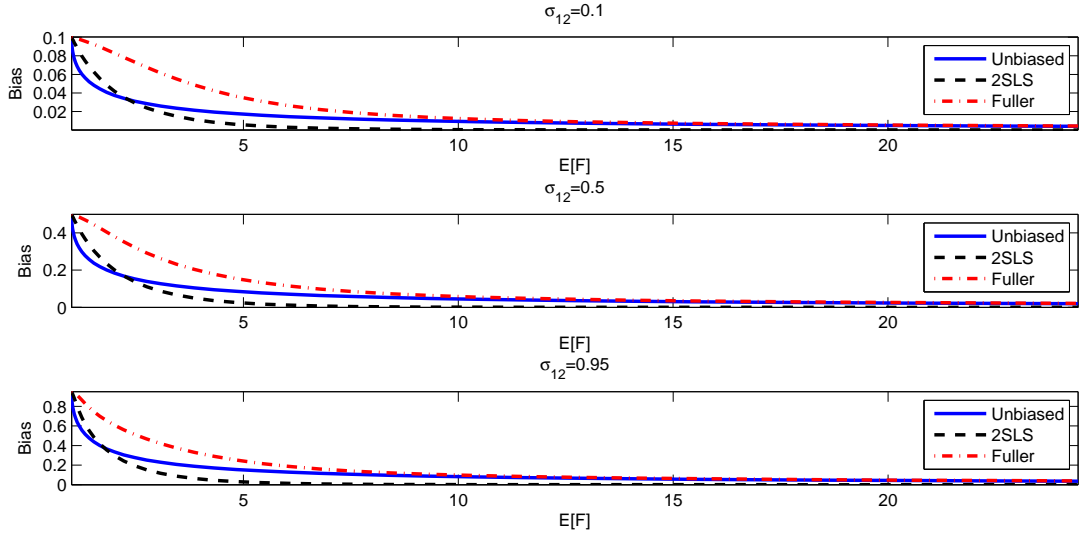}

\protect\caption{Median bias of single-instrument estimators, plotted against mean $E\left[F\right]$
of first-stage F-statistic, based on 10 million simulations.\label{fig Single Medan Instrument Bias}}
\end{figure}

\subsection{Dispersion Simulation Results}\label{dispersion_simulation_append}

The lack of moments for $\hat\beta_{2SLS}$ complicates comparisons of dispersion,
since we cannot consider mean squared error or mean absolute deviation, and also cannot
recenter $\hat\beta_{2SLS}$ around its mean.  As an alternative, we instead consider the full distribution of
the absolute deviation of each estimator from its median.
In particular, for the estimators $(\hat{\beta}_{U},\hat{\beta}_{2SLS},\hat{\beta}_{FULL})$
we calculate the zero-median residuals
\[
\left(\varepsilon_{U},\varepsilon_{2SLS},\varepsilon_{FULL}\right)=\left(\hat{\beta}_{U}-\mbox{med}\left(\hat{\beta}_{U}\right),\hat{\beta}_{2SLS}-\mbox{med}\left(\hat{\beta}_{2SLS}\right),\hat{\beta}_{FULL}-\mbox{med}\left(\hat{\beta}_{FULL}\right)\right).
\]

Our simulation results suggest a strong stochastic ordering between these residuals (in absolute value).
In particular we find that $|\varepsilon_{2SLS}|$ approximately dominates
$|\varepsilon_{U}|$, which in turn approximately dominates $|\varepsilon_{FULL}|,$
both in the sense of first order stochastic dominance. 
This numerical result is consistent with analytical results on the
tail behavior of the estimators. In particular, $\hat{\beta}_{2SLS}$
has no moments, reflecting thick tails in its sampling distribution,
while $\hat{\beta}_{FULL}$ has all moments, reflecting thin tails.
As noted in Section \ref{single_instrument_risk_sec}, the unbiased estimator $\hat{\beta}_{U}$
has a first moment but no more, and so falls between these two
extremes.

To check for stochastic dominance in the distribution of $\left(|\varepsilon_{U}|,|\varepsilon_{2SLS}|,|\varepsilon_{FULL}|\right)$,
we simulated $10^{6}$ draws 
of $\hat{\beta}_{U}$, $\hat{\beta}_{2SLS},$ and $\hat{\beta}_{FULL}$
on a grid formed by the Cartesian product of\\$\sigma_{12}\in\left\{ 0,\left(0.005\right)^{\frac{1}{2}},\left(0.01\right)^{\frac{1}{2}},...,\left(0.995\right)^{\frac{1}{2}}\right\} $
and $\pi\in\left\{ \left(0.01\right)^{2},\left(0.02\right)^{2},...,25\right\} .$
We use these grids for $\sigma_{12}$ and $\pi$, rather than a uniformly
spaced grid, because preliminary simulations suggested that the behavior
of the estimators was particularly sensitive to the parameters for
large values of $\sigma_{12}$ and small values of $\pi$. 

At each point in the grid we calculate $\left(\varepsilon_{U},\varepsilon_{2SLS},\varepsilon_{FULL}\right)$,
using independent draws to calculate $\varepsilon_{U}$ and the other
two estimators, and compute a one-sided Kolmogorov-Smirnov statistic
for the hypotheses that (i) $|\varepsilon_{IV}|\ge|\varepsilon_{U}|$
and (ii) $|\varepsilon_{U}|\ge|\varepsilon_{FULL}|$, where $A\ge B$
for random variables $A$ and $B$  denotes that $A$ is larger than
$B$ in the sense of first-order stochastic dominance. In both cases
the maximal value of the Kolmogorov-Smirnov statistic is less than
$2\times10^{-3}.$  Conventional Kolmogorov-Smirnov p-values
are not valid in the present context (since we use estimated medians
to construct $\varepsilon)$, but are never below 0.25.

\subsection{Containment of $\hat{\beta}_{U}$ in Anderson-Rubin Confidence Set}{\label{beta_U Containment}}

As noted in Section \ref{sec: tests and CS}, the Anderson-Rubin test is uniformly most powerful unbiased in the just identified model.  One can show, however, that the unbiased estimator $\hat{\beta}_{U}$ is not always contained in the Anderson-Rubin confidence set (that is, the confidence set formed by collecting the set of all parameter values not rejected by the Anderson-Rubin test).  Specifically, consider the case where $\xi_2$ is large and negative, $\xi_1$ is large and positive, and $\sigma_{12}$ is non-negative.  In this case, the Anderson-Rubin confidence set will consist solely of negative values, while $\hat{\beta}_{U}$ will be large and positive, and so will necessarily lie outside the Anderson-Rubin confidence set.

While this sort of scenario can easily arise if our sign constraint is violated, it occurs with only low probability when the sign constraint is satisfied.  In particular, as in Section \ref{dispersion_simulation_append} we consider a fine grid of values in the parameter space and simulate the frequency with which the unbiased estimator is contained in the Anderson-Rubin confidence set at each point (based on 100,000 simulations).  We find that the probability that the 95\% Anderson-Rubin confidence set contains the unbiased estimator $\hat{\beta}_{U}$ is always at least 97\%, and exceeds 99.8\% when the mean of the first stage F statistic is greater than two.  Likewise, the probability that the 90\% Anderson-Rubin confidence set contains $\hat{\beta}_{U}$ is always at least 94.5\%, and exceeds 99.3\% when the mean of the first stage F statistic is greater than two.

\section{Multi-Instrument Simulation Design}\label{multi-inst sim design}

This appendix gives further details for the multi-instrument simulation design used in Section \ref{sec: multi-instrument sims}.
We base our simulations on the \citet{staiger_instrumental_1997} specifications for the \citet{angrist_does_1991} data.
The instruments in all specifications are quarter of birth and quarter of birth interacted with other dummy variables, and in all cases the dummy for the fourth quarter (and the corresponding interactions) are excluded to avoid multicollinearity.  The rationale for the quarter of birth instrument in \citet{angrist_does_1991}  indicates that the first stage coefficients on the instruments should therefore be negative.

We first calculate the OLS estimates $\hat\pi$.  All estimated coefficients satisfy the sign restriction in specification I, but some of them violate it in specifications II, III, and IV.  To enforce the sign restriction, we calculate the posterior mean for $\pi$ conditional on the OLS estimates, assuming a flat prior on the negative orthant and an exact normal distribution for the OLS estimates with variance equal to the estimated variance.  This yields an estimate 
$$\tilde\pi_i=\hat\pi_i-\hat\sigma_i\phi\left(\frac{\hat\pi_i}{\hat\sigma_i}\right)/\left(1-\Phi\left(\frac{\hat\pi_i}{\hat\sigma_i}\right)\right)$$
for the first-stage coefficient on instrument $i$,  where $\hat\pi_i$ is the OLS estimate and $\hat\sigma_i$ is its standard error.  When $\hat{\pi_i}$ is highly negative relative to $\hat\sigma_i$, $\tilde\pi_i$ will be close to $\hat\pi_i$, but otherwise  $\tilde\pi_i$ ensures that our first stage estimates all obey the sign constraint.  We then conduct the simulations using $\tilde\pi^*=-\tilde\pi$ to  cast the sign constraint in the form considered in Section \ref{setting_sec}.

Our simulations fix $\tilde\pi^*/\|\tilde\pi^*\|$ at its estimated value and fix $Z'Z$ at its value in the data.   By the equivariance argument in Appendix \ref{equivariance_append} we can fix $\sigma_U^2=\sigma_V^2=1$ and $\beta=0$ in our simulations, so the only remaining free parameters are $\|\pi\|$ and $\sigma_{UV}.$  We consider $\sigma_{UV}\in\{0.1,0.5,0.95\}$ and consider a grid of nine values for $\|\pi\|$ such that the mean of the first stage F statistic varies between 2 and 11.2. For each pair of these parameters we set

$$\Sigma=\left[\begin{array}{cc}
1 & \sigma_{UV}\\
\sigma_{UV} & 1 \end{array}\right]\otimes (Z'Z)^{-1}$$
and  draw of $\xi$ as 

\begin{align*}
\xi\sim N\left(\begin{array}{c}
0  \\
\|\pi\|\cdot  \frac{\tilde\pi^*}{\|\tilde\pi^*\|}\\
\end{array},\Sigma \right).
\end{align*}
\end{document}